\documentclass[lettersize,journal]{IEEEtran}
\usepackage{amsmath,amsfonts}
\usepackage{algorithmicx}
\usepackage{algpseudocode}
\usepackage[linesnumbered,ruled,vlined]{algorithm2e}
\usepackage{array}
\usepackage{textcomp}
\usepackage{stfloats}
\usepackage{url}
\usepackage{verbatim}
\usepackage{graphicx}
\usepackage{subfigure} 
\usepackage{cite}
\usepackage{amssymb}

\usepackage{mathtools}
\usepackage{amsthm}
\usepackage{amstext}
\usepackage{mathrsfs}
\usepackage{multirow}
\usepackage{supertabular}
\usepackage{enumitem}
\usepackage{xcolor}
\usepackage{tabularx}
\usepackage{booktabs}
\usepackage{makecell}
\usepackage{float}  
\usepackage{threeparttable}
\usepackage{orcidlink}
\newtheorem{Theorem}{Theorem}
\newtheorem{Definition}{Definition}

\newtheorem{Lemma}{Lemma}
\allowdisplaybreaks
\begin{document}

\title{Petri Net Modeling and Deadlock-Free Scheduling of Attachable Heterogeneous AGV Systems}

\author{Boyu Li\,\orcidlink{0009-0004-4049-8955},~\IEEEmembership{Graduate Student Member,~IEEE,} Zhengchen Li\,\orcidlink{0009-0004-9515-3655},~\IEEEmembership{Graduate Student Member,~IEEE,} Weimin Wu\,\orcidlink{0000-0002-1958-1920},~\IEEEmembership{Senior Member,~IEEE,} and Mengchu Zhou\,\orcidlink{0000-0002-5408-8752},~\IEEEmembership{Fellow,~IEEE}
\thanks{Boyu Li, Zhengchen Li, and Weimin Wu are with the State Key Laboratory of Industrial Control Technology, Zhejiang University, Hangzhou 310027, China (e-mail: byli@zju.edu.cn; 12432031@zju.edu.cn; wmwu@iipc.zju.edu.cn).}
\thanks{MengChu Zhou is with the School of Information and Electronic Engineering, Zhejiang Gongshang University, Hangzhou 310018, China, and also with the Department of Electrical and Computer Engineering, New Jersey Institute of Technology, Newark, NJ 07102 USA (e-mail: mengchu@gmail.com).}
\thanks{This work has been submitted to the IEEE for possible publication. Copyright may be transferred without notice, after which this version may no longer be accessible.}
}


\maketitle

\begin{abstract}
	The increasing demand for flexible automation has accelerated the adoption of heterogeneous automated guided vehicles (AGVs).
	This work investigates a new scheduling problem in a material transportation system consisting of attachable heterogeneous AGVs, including carriers and shuttles, that flexibly attach and detach for cooperative task execution.
	While such collaboration enhances operational efficiency, the attachment-induced synchronization renders the system highly coupled and susceptible to deadlocks.
	To address this, we propose a Petri net (PN)-based deadlock-free scheduling framework integrated into an adaptive large neighborhood search (ALNS) algorithm.
	The PN is introduced to map candidate solutions from static permutations into dynamic collaborative processes, enabling performance evaluation via state evolution and proactive deadlock prevention through structural analysis.
	Extensive experiments on real-world and synthetic instances demonstrate that the proposed framework significantly improves computational efficiency, with the developed ALNS outperforming the current on-site policy, exact solvers, and state-of-the-art metaheuristics.
	Finally, sensitivity analysis yields managerial insights for optimal fleet sizing.

\end{abstract}

\def\abstractname{Note to Practitioners}
\begin{abstract}
	Motivated by a real-world heterogeneous AGV application, this paper addresses the practical scheduling challenge of a carrier–shuttle system, where two types of physically attachable AGVs with complementary functions collaborate to perform transportation tasks.
	However, such tight operational dependence imposes synchronization constraints, rendering the system highly susceptible to deadlocks and computationally challenging.
	To resolve this, a PN-enhanced deadlock-free ALNS framework is proposed.
	By explicitly modeling the coupled precedence relations and synchronization logic, the PN enables precise performance assessment through simulation and proactively prevents deadlocks via structural analysis.
	Experimental results demonstrate that the proposed deadlock-free framework significantly accelerates the metaheuristic search process, and the ALNS consistently yields superior scheduling outcomes compared with competing methods.
	Finally, our operational analyses identify the carrier as the primary bottleneck, yielding managerial insights for cost-effective fleet configuration.
	While developed for a specific carrier--shuttle system, the proposed framework is extensible to broader heterogeneous multi-robot collaboration with task-level synchronization requirements.
	Nevertheless, the proposed framework focuses on scheduling-level decisions under deterministic travel and operation times.
	Execution uncertainties, including motion control errors, communication delays, and AGV faults, are not explicitly considered.
	
\end{abstract}


\section{Introduction}

Rapid advances in automation and robotics have accelerated the adoption of automated guided vehicles (AGVs) in logistics, manufacturing, and warehousing.
By automating material transportation, AGVs can improve productivity, reduce labor costs, and enhance operational safety.
With the increasing demand for full automation, heterogeneous AGVs attract significant attention from both academia and industry.
Inspired by the practical industrial context depicted in Fig. \ref{factory}, this work aims to study a new attachable heterogeneous AGV scheduling problem (AHASP) in a battery manufacturing system.

As shown in Fig. \ref{factory}, two types of heterogeneous AGVs with complementary capabilities collaborate to transport battery plates in a manufacturing system.
1) \textit{Shuttles} with material handling capabilities navigate confined production chambers to load and unload plates, and are restricted to local movement within chambers.
2) \textit{Carriers} with substantial load capacity and global transport capabilities enable the inter-chamber transfer of loaded or unloaded shuttles.
As illustrated in Fig.~\ref{workflow}, this cooperative workflow is structured into four sequential stages: \textit{Attach}→\textit{Retrieval}→\textit{Delivery}→\textit{Detach}. 
In this process, the carrier and shuttle attach to enable synchronized inter-station transport and detach to operate independently.
Therefore, each task necessitates synchronized collaboration between a carrier and a shuttle, with the carrier relying on the shuttle for material handling and the shuttle depending on the carrier for transfer.
This physical attachment is motivated by practical constraints in the targeted battery workshop, which involves transporting heavy battery plates and handling them within tightly confined chambers.
In this setting, an all-in-one AGV that integrates high-payload transport with intra-chamber handling is technically complex and costly, potentially compromising reliability.
By contrast, the carrier-shuttle system enables functional decoupling and parallel operations, offering a more cost-effective, efficient, and reliable solution.

Despite these advantages, attachment-induced collaboration poses unique and severe challenges for scheduling.
From a combinatorial optimization perspective, AHASP is NP-hard as it generalizes the classical Vehicle Routing Problem (VRP) \cite{dohn2011vehicle}.
However, beyond standard routing, AHASP necessitates the joint determination of carrier–shuttle pairing, task assignment, and AGV execution sequencing, which substantially expands the solution space.
To tackle this complexity, we employ an Adaptive Large Neighborhood Search (ALNS) framework, which has been proven effective in exploring the large and highly coupled search space \cite{voigt2024review}.
However, the application of such metaheuristics to AHASP entails distinct challenges regarding solution encoding-decoding and feasibility assurance.
Standard permutation-encoded solutions typically capture only independent routing sequences, failing to characterize the tight cross-schedule dependencies induced by synchronization.
Thus, a gap arises when decoding and realizing static solutions into the dynamic collaborative execution process for performance evaluation.
More critically, the resulting coupled schedules involve cascading precedence relations that are highly prone to circular waiting for resources, leading to deadlocked and thus infeasible solutions.
The frequent generation of such infeasible candidates wastes substantial computational effort and slows convergence of the metaheuristic search.

To address this, we introduce Petri nets (PNs) as a powerful formalism widely used to represent discrete-event systems with concurrency, synchronization, and sequentiality, all of which are inherent in the execution process of AHASP solutions.
Moreover, PNs provide a rigorous theoretical foundation for deadlock analysis and resolution.
Unlike traditional studies that primarily employ PNs for system-level modeling to design general control policies, we construct PNs from a novel solution-level modeling perspective to represent the candidate solutions generated by ALNS.
The PN model functions as a dynamic simulator to transform static permutation-based solutions into explicit execution processes, enabling precise performance evaluation based on state evolution.
Crucially, deadlocks are proactively prevented via PN structural analysis, ensuring deadlock-free schedules.
The main contributions of this paper are summarized as follows:

\begin{figure}[t]
	\centering
	\includegraphics[width=0.80\linewidth]{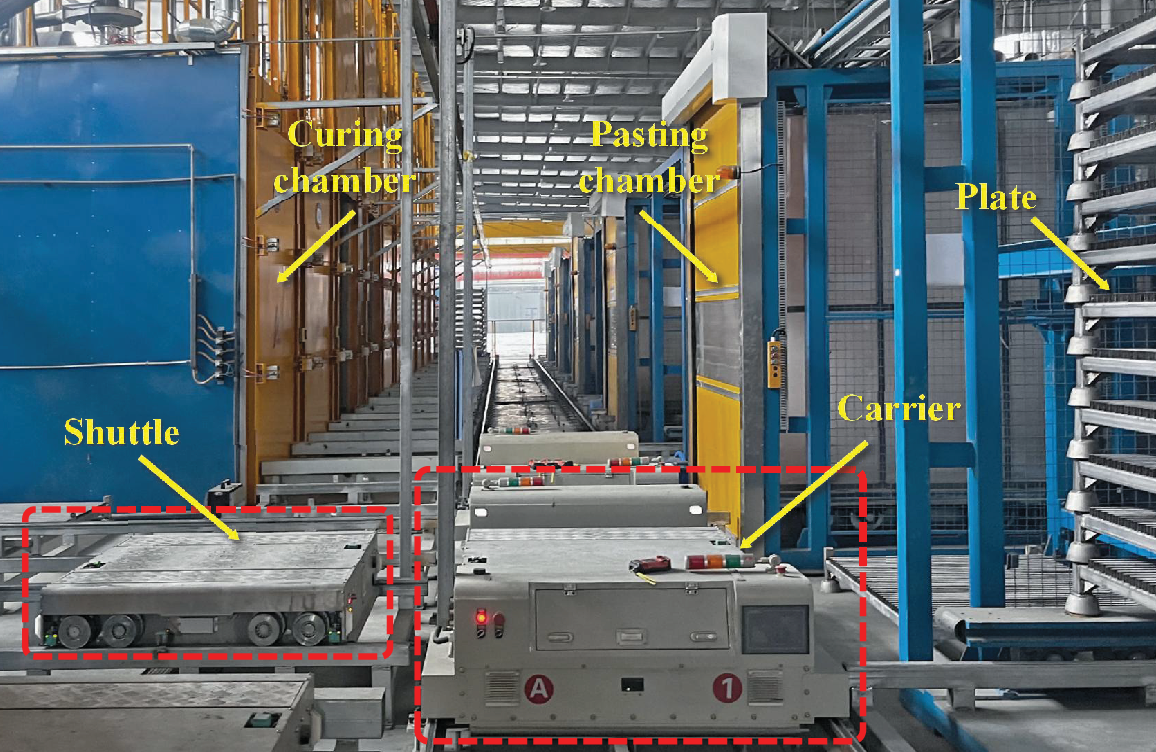}
	\caption{Battery manufacturing system in Jiangxi, China.}
	\label{factory}
	\vspace{-10pt}
\end{figure}

\begin{itemize}
	\item We abstract AHASP from real-world engineering scenarios and establish a mixed-integer linear programming (MILP) formulation to formally define the problem.

	\item We propose a PN-enhanced deadlock-free scheduling framework integrated into the ALNS algorithm.
	By constructing solution-level PNs to explicitly characterize coupled precedence relations and synchronization logic, we develop a Firing-Driven Decoding (FDD) method for precise evaluation and a Bidirectional Dependency Search (BDS) strategy for proactive deadlock prevention.

	\item We conduct extensive experiments on both real-world industrial and synthetic datasets.
	The results validate that the proposed PN-based deadlock-free framework significantly improves both computational efficiency and solution quality, and that the developed ALNS outperforms the on-site policy, two exact solvers, and four state-of-the-art metaheuristics.

\end{itemize}

The remainder of the paper is organized as follows.
Section \ref{Literature_review} reviews related work.
Section \ref{Problem description} formulates AHASP as a mathematical model.
Section \ref{Petri_feasibility} constructs the PN model and analyzes deadlocks.
Section \ref{Adaptive Large Neighborhood Search} develops the PN-enhanced metaheuristic.
Section \ref{Numerical Experiment} presents experimental results.
Section \ref{Conclusion} concludes this paper.

\section{Related Work} \label{Literature_review}

\subsection{Heterogeneous Robot Scheduling Problem (HRSP)}

Fu et al. \cite{fu2022robust} classify robot heterogeneity into two categories: structural (e.g., size, speed, and capacity) and functional (e.g., loading, transporting, and stacking).
A representative structural HRSP is the heterogeneous vehicle routing problem (HVRP), involving vehicles with varying costs and capacities \cite{pereira2024approach}.
Despite its various forms, structural heterogeneity is typically reflected through cost and capacity constraints \cite{Ba2026Multi}.
In the absence of collaboration and synchronization, HVRP tasks are scheduled independently, which resembles homogeneous MRTA.
The AHASP addressed in this paper involves functionally heterogeneous AGVs operating collaboratively with synchronization constraints.

\begin{figure}[t]
	\centering
	\includegraphics[width=0.85\linewidth]{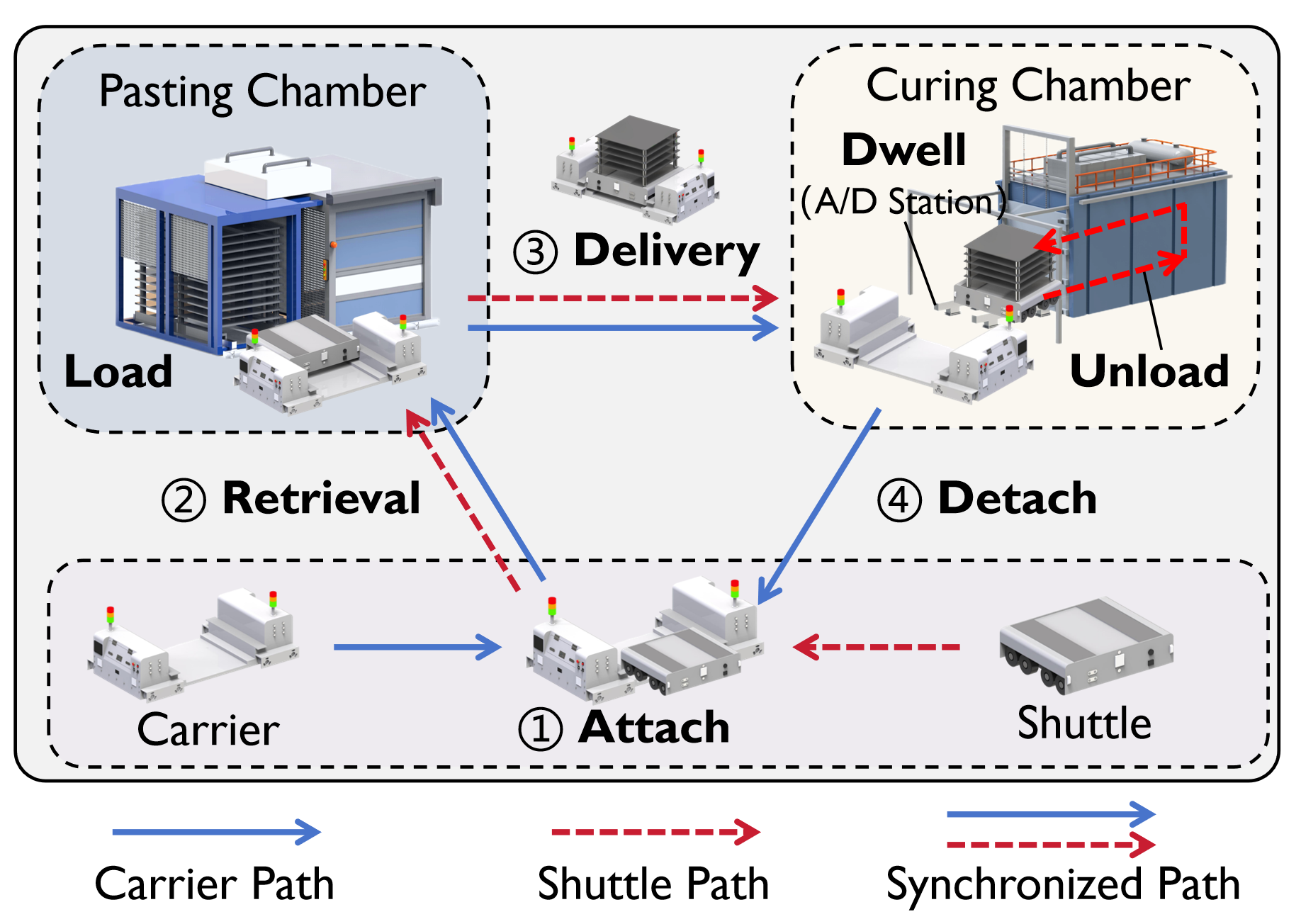}
	\caption{Cooperative carrier–shuttle workflow for task execution in an attachable AGV system.}
	\label{workflow}
	\vspace{-10pt}
\end{figure}

Functionally heterogeneous robots possess distinct yet complementary capabilities and cooperate to accomplish compound tasks \cite{martin2025collaborative,bai2019efficient,deng2024distributed,zuo2026self}.
A common example is the cooperation between forklifts and transport vehicles to transport and load cargo onto warehouse storage racks \cite{wang2024task,li2024efficient}.
This cooperation synchronizes their schedules and exhibits their interdependence whereby the utilization of one robot directly affects or depends on the scheduling of others \cite{gombolay2018fast}.
A widely adopted strategy for collaborative HRSP is to decompose compound tasks into elemental tasks with interrelated constraints.
Each elemental task is then assigned to a robot with the required functionality \cite{kamra2017combinatorial,chen2020yard,fu2022robust,wang2020coupled,ferreira2024distributed}.
Specifically, Chen et al. \cite{chen2020yard} investigate the integrated scheduling of yard cranes and AGVs, decomposing it into a set of crane-specific and vehicle-specific sub-tasks.
Ferreira et al. \cite{ferreira2024distributed} examine the scheduling of heterogeneous multi-robot teams, breaking down compound tasks into fine-grained action-level components.
Fu et al. \cite{fu2022robust} develop a stochastic programming framework for heterogeneous multi-agent systems that concurrently optimizes decomposition, task assignment, and scheduling.

Although AHASP involves functionally heterogeneous AGVs, it introduces a unique attachment-induced physical dependency between carriers and shuttles, which is absent in existing HRSP studies \cite{fu2022robust}.
This dependency renders existing decomposition strategies inapplicable to AHASP, as the compound task in AHASP cannot be decomposed into carrier-specific and shuttle-specific sub-tasks for separate assignment.
A similar attachment-based dependency is observed in unmanned air–ground vehicle (UAV/UGV) systems where UAVs depend on UGVs for transportation, charging, and deployment to conduct aerial surveillance \cite{wu2025bin,zhou2026heterogeneous}.
Existing research on UAV/UGV systems typically favors one side or partitions cooperative scheduling into UGV and UAV-level sub-problems, solved hierarchically \cite{li2021memetic,li2025two}.
While two-level scheduling reduces computational complexity, it weakens collaboration and sacrifices global optimality.
Coordinated scheduling for a non-decomposable carrier-shuttle system remains a new and challenging problem.

\subsection{Scheduling with Synchronization Constraints}

Synchronization constraints are prevalent in various cooperative scheduling domains where robots must collaborate simultaneously to execute tasks \cite{soares2024synchronisation}.
These constraints are first introduced by Eveborn et al. \cite{eveborn2006laps} within the context of staff scheduling in healthcare systems where two caregivers are required to serve a single patient simultaneously.
Bredstr{\"o}m et al. \cite{bredstrom2008combined} further extend synchronization constraints into temporal precedence and investigate synchronization issues in contexts such as forest operations and homecare staff scheduling.
Dohn et al. \cite{dohn2011vehicle} then generalize synchronization and precedence constraints as temporal dependencies.
Such synchronized scheduling problems appear in diverse real-world scenarios, including the delivery of medication, appliances, and electrical equipment \cite{erdogan2023synchronizing}, two-echelon distribution \cite{lehmann2024matheuristic}, staff scheduling \cite{sartori2022scheduling}, and cooperative multi-robot missions with task precedence relationships \cite{gosrich2025online}.
To model the synchronized precedence relations in these problems, two common approaches are precedence graph \cite{masson2013efficient,rist2024benders} and precedence matrix \cite{liu2019adaptive,wang2023adaptive}.
Specifically, Masson et al. \cite{masson2013efficient} first introduce the precedence graph, also referred to as the constraint graph \cite{rist2024benders}, which is widely used to identify inconsistencies.
Liu et al. \cite{liu2019adaptive} formulate dependencies using a precedence matrix and utilize it to prevent cross-synchronization.
Wang et al. \cite{wang2023adaptive} combine the precedence graph and matrix to assess both time-related and cycle-related feasibility.
However, these approaches rely on static and stateless representations, failing to capture the dynamic state evolution and synchronization behaviors arising during schedule execution.

In contrast, PNs provide a powerful formalism for modeling discrete-event systems by extending static precedence through markings, enabling the explicit representation of system dynamics involving concurrency, synchronization, and sequences \cite{david1994petri}.
Leveraging these capabilities, PNs have been widely applied to schedule resource-constrained operations and resolve deadlock issues \cite{yi2024heuristic, Lee2025Optimal, shi2026Allocation, li2024State}.
Despite their strong suitability, the application of PNs to scheduling heterogeneous robots with synchronization constraints remains limited.
Moreover, most existing PN-based scheduling studies focus on modeling overall system behavior to derive optimal schedules.
In contrast, from a novel solution-level modeling perspective, this work constructs PNs to capture the dynamic execution of coupled schedules and achieve deadlock-free scheduling for AHASP.

\section{Problem description and formulation}\label{Problem description}
\subsection{Problem Description} \label{introduce flex}
\begin{figure}[ht]
	\centering
	\includegraphics[width=1\linewidth]{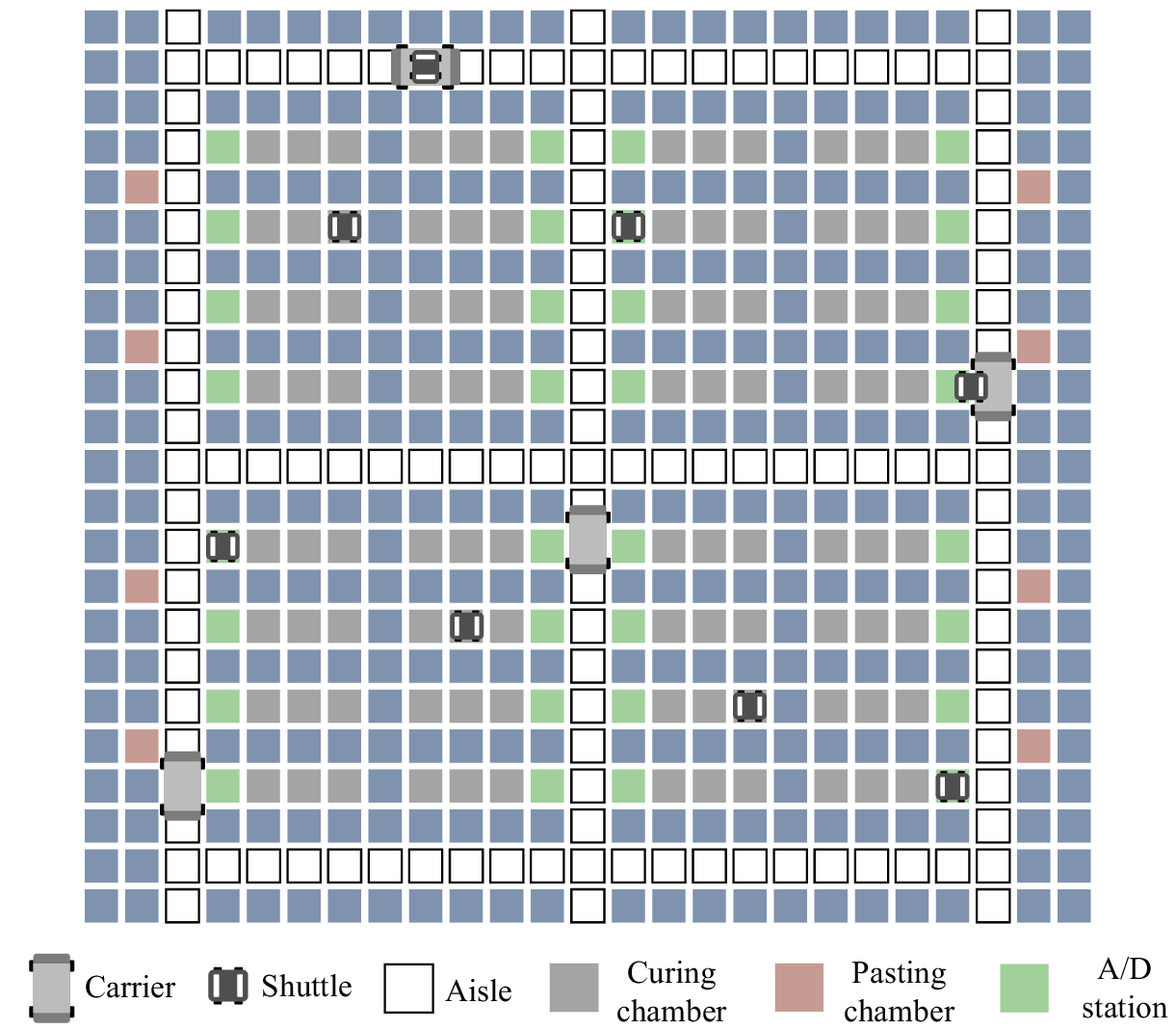}
	\caption{Layout of the battery manufacturing system.}
	\label{layout}
	\vspace{-10pt}
\end{figure}

\begin{figure}[ht]
	\centering
	\includegraphics[width=1\linewidth]{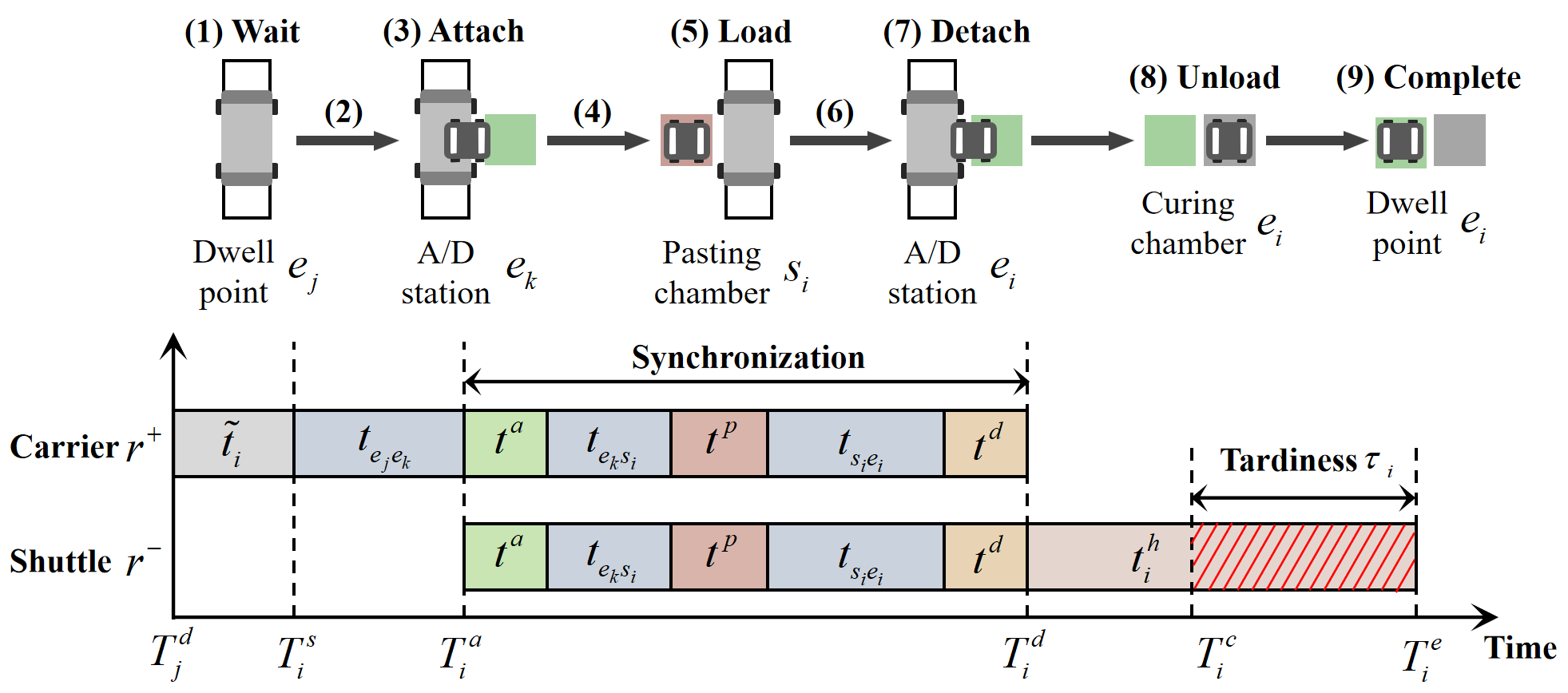}
	\caption{Gantt chart illustrating the synchronized execution of task $i$ by carrier $r^+$ and shuttle $r^-$.}
	\label{execution_process}
	\vspace{-10pt}
\end{figure}

The layout of the battery manufacturing system is shown in Fig. \ref{layout}.
Battery plates are pasted in pasting chambers and then transported to curing chambers via carriers and shuttles.
Each task $i$ involves a pasting chamber as source position $s_i$ and a curing chamber as destination position $e_i$, requiring the cooperation between a carrier $r^+$ and a shuttle $r^-$.
Let tasks $j$ and $k$ be the preceding tasks of carrier $r^+$ and shuttle $r^-$, respectively, before executing task $i$.
The execution process of task $i$ is shown in Fig. \ref{execution_process} and detailed as follows.
\begin{enumerate}
	\item Upon completing task $j$ at time $T_j^\textnormal{d}$, carrier $r^+$ dwells at position $e_j$ for duration $\tilde{t}_i$ to synchronize with the availability of shuttle $r^-$;
	\item Carrier $r^+$ starts task $i$ at time $T_i^\textnormal{s}$, moving from position $e_j$ to the location of shuttle $r^-$, i.e., $e_k$;
	\item Shuttle $r^-$ begins attaching to carrier $r^+$ at time $T_i^\textnormal{a}$;
	\item Carrier $r^+$ transfers shuttle $r^-$ to pasting chamber $s_i$;
	\item Shuttle $r^-$ loads pasted plates;
	\item Carrier $r^+$ transfers shuttle $r^-$ to curing chamber $e_i$;
	\item Shuttle $r^-$ detaches from carrier $r^+$, and carrier $r^+$ completes task $i$ at time $T_i^\textnormal{d}$;
	\item Shuttle $r^-$ unloads plates in curing chamber $e_i$;
	\item Shuttle $r^-$ completes task $i$ and dwells at the attach/detach (A/D) station of $e_i$ at time $T_i^\textnormal{e}$.
\end{enumerate}

As shown in Fig. \ref{execution_process}, carrier $r^+$ and shuttle $r^-$ initiate and complete task $i$ asynchronously, but operate synchronously during the attachment period. 
To ensure plate quality, each task $i$ has a designated curing due time $T_i^\textnormal{c}$, and ideally the task completion time $T_i^\textnormal{e}$ should be earlier than $T_i^\textnormal{c}$.
However, due to limited transport capacity and physical travel constraints, satisfying all deadlines is not always feasible, particularly under heavy workloads.
Consequently, a penalty cost is incurred for the tardiness (i.e., $\max(0, T_i^\textnormal{e} - T_i^\textnormal{c})$), as indicated by the red slanted bar in Fig. \ref{execution_process}.
The paths of carriers are planned by using the A$^*$ algorithm, with the travel distance contributing to the overall cost.
The travel time $t_{a, b}$ between two positions $a$ and $b$ is calculated as $t_{a, b} = d_{a, b} / v$, where $d_{a, b}$ represents the path distance and $v$ is the constant velocity of carriers.
Given the shuttle’s limited intra-chamber movement, its travel cost is considered negligible.
To establish a clear baseline and focus on the core synchronization constraints and deadlock prevention, this work assumes deterministic travel time and fault-free AGVs as a baseline setting.

\subsection{Mathematical Formulation} \label{Mathematical_formulation}

\begin{table}[h]
	\centering
	\small
	\begin{tabularx}{\textwidth}{c X}
		\multicolumn{2}{l}{\textbf{Parameters}}\\
		$\mathcal{N}$       & Set of tasks \\
		$\mathcal{N}_o$     & Set of virtual tasks \\
		$\mathcal{N'}$ & Set of nodes \\
		$\mathcal{V}$       & Set of all AGVs \\
		$\mathcal{V}^+$     & Set of carriers \\
		$\mathcal{V}^-$     & Set of shuttles \\
		$n$       & Number of tasks \\
		$n_r$     & Number of tasks assigned to AGV $r$ \\
		$m^+$     & Number of carriers \\
		$m^-$     & Number of shuttles \\
		$o_r$     & Virtual task of AGV $r$ \\
		$s_i$     & Source position of task $i$ \\
		$e_i$     & Destination position of task $i$ \\
		$v$       & Velocity of AGVs \\
		
		$t^\textnormal{a}$      & Attachment duration \\
		$t^\textnormal{d}$      & Detachment duration \\
		$t^\textnormal{p}$      & Load duration \\
		$t_i^\textnormal{h}$    & Unload duration of task $i$ \\
		
		$T_i^\textnormal{s}$    & Start time of task $i$ \\
		$T_i^\textnormal{a}$    & Attach time of task $i$ \\
		$T_i^\textnormal{c}$    & Designated curing due time of task $i$ \\ 
		$\lambda$    & Cost weight factor \\
		$L$          & A sufficiently large positive value \\
	\end{tabularx}
	\vspace{-15pt}
\end{table}


\begin{table}[h]
	\centering
	\small
	\begin{tabularx}{\textwidth}{c X}
		\multicolumn{2}{l}{\textbf{Decision variables}} \\
		$x_{ij}^r$       & \makecell[l]{$x_{ij}^r = 1$ if edge $(i,j)$ is assigned to AGV $r$ and \\ 0 otherwise} \\ 
		$\tilde{t}_i$  & Carrier's synchronization dwell time for task $i$ \\ 
		$T_i^\textnormal{d}$    & Detach time of task $i$ \\
		$T_i^\textnormal{e}$    & Completion time of task $i$ \\
		$\delta_i$   & Travel distance of task $i$ \\
		$\tau_i$     & Tardiness of task $i$ \\
	\end{tabularx}
\end{table}

AHASP is represented as a directed graph $\mathcal{G}=\{ \mathcal{N'}, \mathcal{E}\}$, where $\mathcal{N'} = \mathcal{N}_o \cup \mathcal{N}$ is the set of nodes and $\mathcal{E} = \{(i,j) \mid i,j \in \mathcal{N'}, i \neq j\}$ is the set of edges.
The set $\mathcal{N}_o = \{o_r \mid r \in \mathcal{V}\}$ defines a virtual depot $o_r$ for each AGV $r$, from which its task sequence starts and ends, with $e_{o_r}$ denoting its initial position.
The nodes in $\mathcal{N}= \{1,2,\cdots,n\}$ correspond to transportation tasks.
AHASP entails assigning each task in $\mathcal{N}$ to a carrier in $\mathcal{V}^+ = \{1, 2, \cdots, m^+\}$ and a shuttle in $\mathcal{V}^- = \{m^+ + 1, m^+ + 2, \cdots, m^+ + m^-\}$, and determining the execution sequence for each AGV in $\mathcal{V} = \mathcal{V}^+ \cup \mathcal{V}^-$.
The objective is to minimize the weighted sum of travel distance and tardiness.
The MILP model for AHASP is formulated as follows to provide a mathematical definition of the problem.

\begin{gather}
	\min F= \lambda \cdot \sum_{i \in \mathcal{N}} \delta_i + (1 - \lambda) \cdot \sum_{i \in \mathcal{N}} \tau_i \label{aim}
\end{gather}
s.t.:
\begin{gather}
	\sum_{r \in \mathcal{V}^+}\sum_{i \in \mathcal{N'}} x_{ij}^r = 1, \forall j \in \mathcal{N} \label{c1_1} \\
	\sum_{r \in \mathcal{V}^-}\sum_{i \in \mathcal{N'}} x_{ij}^r = 1, \forall j \in \mathcal{N} \label{c1_2} \\
	\sum_{r \in \mathcal{V}^+}\sum_{j \in \mathcal{N'}} x_{ij}^r = 1, \forall i \in \mathcal{N} \label{c1_3} \\
	\sum_{r \in \mathcal{V}^-}\sum_{j \in \mathcal{N'}} x_{ij}^r = 1, \forall i \in \mathcal{N} \label{c1_4} \\
	\sum_{i \in \mathcal{N'}} x_{ij}^r - \sum_{i \in \mathcal{N'}} x_{ji}^r = 0, \forall j \in \mathcal{N}, r \in \mathcal{V} \label{c2}  \\
	\sum_{i\in \mathcal{N'}} x_{io_r}^r = \sum_{i\in \mathcal{N'}} x_{o_r i}^r = 1, r \in \mathcal{V}  \label{c3}\\ 
	d_{e_{j}, e_{k}} + d_{e_{k} ,s_i} + d_{s_i, e_i} - L(2 - x_{j i}^{r^+} - x_{k i}^{r^-}) \leq \delta_i, \nonumber \\ 
	\forall j ,k \in \mathcal{N'}, i \in \mathcal{N} , r^+ \in \mathcal{V}^+, r^- \in \mathcal{V}^- \label{distance_task} \\
	T_k^\textnormal{e} - T_j^\textnormal{d} - t_{e_{j}, e_{k}} - L (2 - x_{j i}^{r^+} - x_{k i}^{r^-}) \leq \tilde{t}_i, \nonumber \\ 
	\forall j,k \in \mathcal{N'}, i \in \mathcal{N}, r^+ \in \mathcal{V}^+, r^- \in \mathcal{V}^- \label{idle_1} \\
	T_j^\textnormal{d} + \tilde{t}_i + t_{e_{j}, e_{k}} + t^\textnormal{a} + t_{e_{k}, s_{i}} + t^\textnormal{p} + t_{s_i, e_i} + t^\textnormal{d} - L (2 - x_{j i}^{r^+} - \nonumber \\ 
	x_{k i}^{r^-}) \leq T_i^\textnormal{d}, \forall j,k \in \mathcal{N'}, i \in \mathcal{N}, r^+ \in \mathcal{V}^+, r^- \in \mathcal{V}^- \label{s_e} \\
	T_i^\textnormal{d} + t_i^\textnormal{h} \leq T_i^\textnormal{e}, \forall i \in \mathcal{N}  \label{e_d} \\
	T_i^\textnormal{e} - T_i^\textnormal{c} \leq \tau_i, \forall i \in \mathcal{N}  \label{tardiness_1} \\
	x_{ij}^r = 0 , \forall i,j \in \mathcal{N}, r \in \mathcal{V},i = j \label{c11} \\
	\tilde{t}_i, \tau_i, T_i^\textnormal{d}, T_i^\textnormal{e} \geq 0, \forall i \in \mathcal{N'} \label{c5_1} \\
	x_{ij}^r \in \lbrace 0,1\rbrace, \forall i,j \in \mathcal{N'}, r \in \mathcal{V} \label{c10} 
\end{gather}

The objective function is defined in \eqref{aim}.
Constraints \eqref{c1_1}-\eqref{c1_4} guarantee the assignment of each task to one carrier and one shuttle.
Constraint \eqref{c2} enforces flow conservation.
Constraint \eqref{c3} ensures that the task sequence of AGV $r$ starts and ends at virtual task $o_r$.
Constraint \eqref{distance_task} determines the travel distance of each task.
Constraint \eqref{idle_1} determines the carrier's dwell time to synchronize its arrival with the shuttle's availability.
Constraint \eqref{s_e} defines the relationship between the detachment times of two consecutive tasks.
Constraint \eqref{e_d} establishes the relationship between the detachment time and completion time of each task.
Constraint \eqref{tardiness_1} determines the tardiness of each task.
Constraint \eqref{c11} eliminates self-visits.
Constraints \eqref{c5_1} and \eqref{c10} define the feasible domains of decision variables.

\section{Petri Net Modeling and Deadlock Analysis} \label{Petri_feasibility}

Given the NP-hard nature of AHASP, the MILP formulation is primarily used for problem definition and small-scale validation, as exact MILP optimization is computationally prohibitive for industry-scale instances.
To tackle this computational complexity, we employ the ALNS framework to efficiently obtain high-quality near-optimal solutions.
However, applying ALNS to the tightly coupled AHASP under synchronization constraints is nontrivial in terms of solution encoding-decoding and feasibility enforcement.
Specifically, standard permutation-encoded solutions fail to characterize the cross-schedule dependencies and synchronization logic, thereby preventing the direct decoding of static sequences into the dynamic collaborative execution process for performance evaluation.
More critically, the resulting coupled schedules involve cascading precedence relations that are highly prone to deadlocks, leading to infeasible solutions.

To address these challenges, we introduce a PN-based deadlock-free scheduling framework integrated into ALNS to explicitly model candidate solutions, thereby enabling precise performance evaluation and steering the search towards feasible regions.
The overall architecture of the proposed framework is illustrated in Fig.~\ref{fig:framework}.
In this section, we first detail \textbf{PN Modeling} (left) to map permutations into PN representations, and \textbf{PN-based Analysis} (middle), where FDD simulates the PN state evolution to compute the objective function, while BDS prevents deadlocks during construction to ensure feasibility.
Subsequently, \textbf{Metaheuristic} (right) is developed based on the ALNS framework to leverage these PN mechanisms to achieve deadlock-free scheduling.

\begin{figure}[t]
	\centering
	\includegraphics[width=0.99\linewidth]{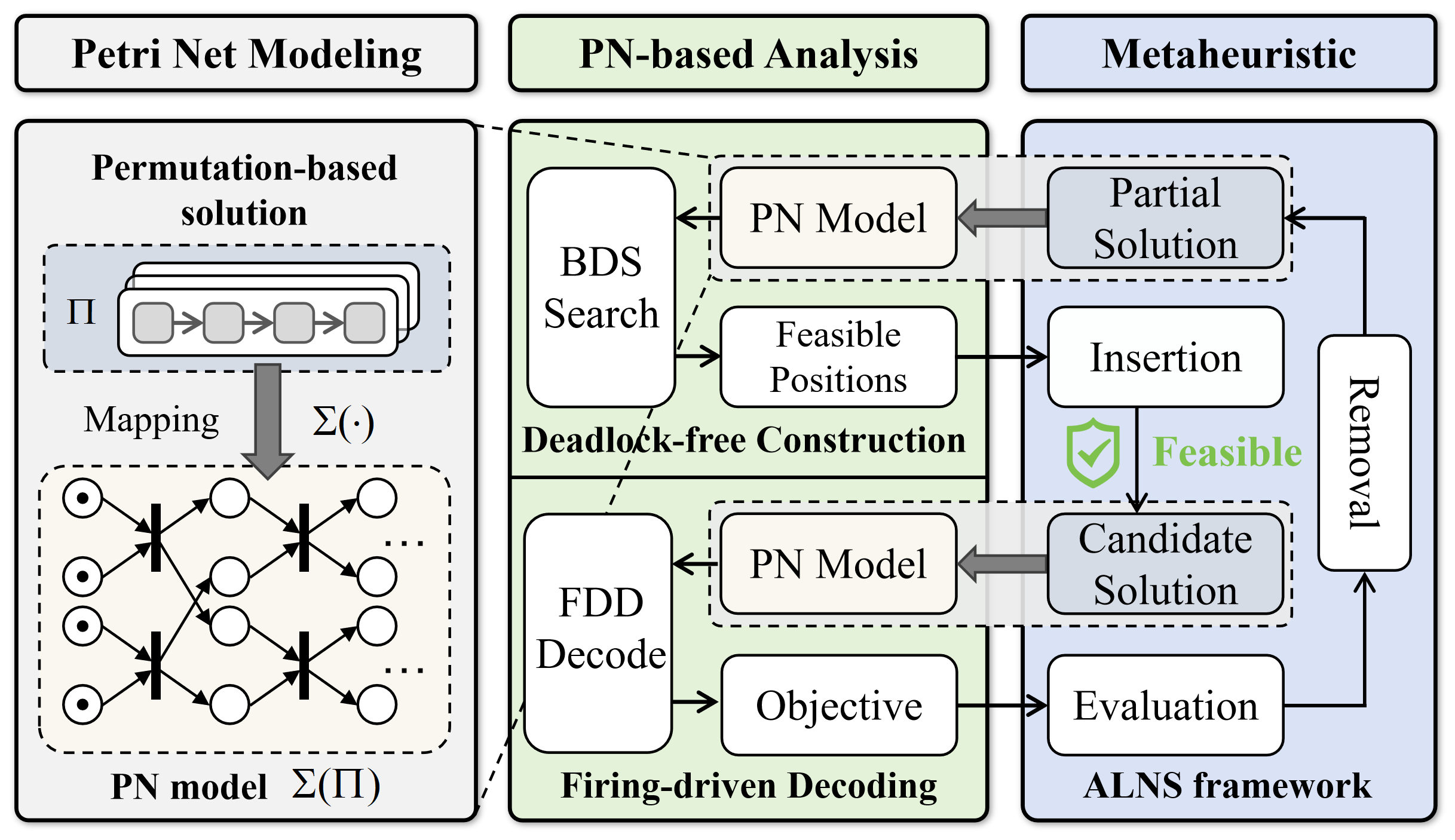}
	\caption{Architecture of the proposed PN-enhanced ALNS framework for deadlock-free scheduling.}
	\label{fig:framework}
	\vspace{-10pt}
\end{figure}

\subsection{Permutation-based Solution Encoding}\label{Solution representation}

We begin by introducing the standard permutation-based solution encoding utilized in the ALNS framework, which serves as the input for the subsequent PN modeling. 
This encoding represents a solution as a dual-vector structure $\Pi = (\Pi^+, \Pi^-)$, where $\Pi^+ = (\pi_r)_{r\in\mathcal{V}^+}$ and $\Pi^- = (\pi_r)_{r\in\mathcal{V}^-}$ are vectors of finite task sequences for carriers and shuttles, respectively.
For each AGV $r \in \mathcal{V}$, its route is represented as a finite sequence of tasks $\pi_r = (0, \pi_{1}^r, \pi_{2}^r, \cdots ,\pi_{n_{r}}^r, 0)$, where $n_r$ denotes the number of tasks assigned to AGV $r$, and element $0$ acts as a placeholder for the specific virtual depot node $o_r$.
The inclusion of depot node $0$ at both route endpoints allows us to decompose the entire trajectory into a set of adjacency pairs, from the initial departure $(0, \pi_1^r)$ through intermediate tasks $(\pi_k^r, \pi_{k+1}^r)$ to the final return $(\pi_{n_r}^r, 0)$.
To formalize the precedence relations, we define notation $(i, j) \in \pi_r$ to indicate that node $i$ immediately precedes node $j$ in sequence $\pi_r$. 
By extension, $(i, j) \in \Pi^+$ (resp. $\Pi^-$) denotes that adjacency relation $(i, j)$ exists in one of the carrier (resp. shuttle) routes, i.e., $\exists r \in \mathcal{V}^+$ (resp. $\mathcal{V}^-$) such that $(i, j) \in \pi_r$.
For cooperative execution, every task $i \in \mathcal{N}$ appears exactly once in carrier routes $\Pi^+$ and exactly once in shuttle routes $\Pi^-$.

Consider an illustrative example with seven tasks $\mathcal{N} = \{1,\cdots,7\}$ assigned to two carriers $\mathcal{V}^+ = \{1,2\}$ and two shuttles $\mathcal{V}^- = \{3,4\}$.
The permutation-encoded solution $\Pi$ specifies the following task sequences:
\begin{enumerate}
	\item Carrier 1: $\pi_{1} = ( 0 \to 5 \to 4 \to 2 \to 0 )$
	\item Carrier 2: $\pi_{2} = ( 0 \to 6 \to 1 \to 3 \to 7 \to 0 )$
	\item Shuttle 3: $\pi_{3} = ( 0 \to 5 \to 1 \to 4 \to 7 \to 0 )$
	\item Shuttle 4: $\pi_{4} = ( 0 \to 6 \to 3 \to 2 \to 0 )$
\end{enumerate}

In this example, a direct precedence relation $5 \prec 4$ is defined locally in $\pi_1$.
Notably, a cross-route precedence $5 \prec 3$ emerges implicitly, as task 5 precedes task 1 in $\pi_3$, whereas task 1 precedes task 3 in $\pi_2$.
This propagation through shared task 1 exemplifies cascading precedence relations.
Such dependencies tend to be deadlock-prone, although this specific instance is valid.
This instance serves as a running example throughout the subsequent sections to illustrate PN modeling, decoding procedures, and deadlock prevention.

\subsection{Petri Net Modeling} \label{Petri net modeling}

\begin{figure}[t]
	\centering
	\includegraphics[width=0.95\linewidth]{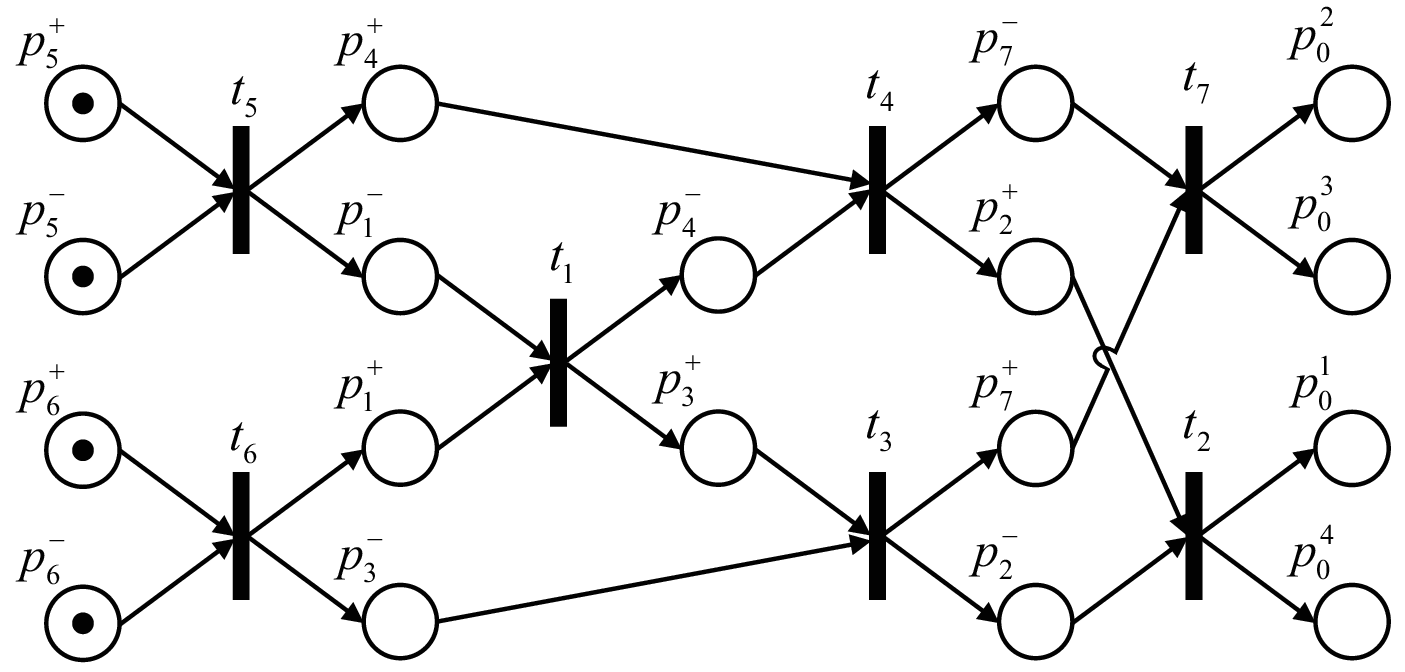}
	\caption{PN $\Sigma(\Pi)$ modeling the example solution $\Pi$.}
	\label{solution_petri}
	\vspace{-10pt}
\end{figure}

Corresponding to the left panel of Fig. \ref{fig:framework}, this subsection details the PN modeling that maps static permutation $\Pi$ into an executable PN model, explicitly characterizing the synchronization constraints and cross-schedule dependencies.

Formally, we adopt Ordinary Petri Nets, defined as a five-tuple $\Sigma = \langle \mathcal{P}, \mathcal{T}, Pre, Post, M_0 \rangle$, where $\mathcal{P}$ and $\mathcal{T}$ are disjoint finite sets of places and transitions, respectively.
The structure is defined by the pre-incidence function $Pre: \mathcal{P} \times \mathcal{T} \rightarrow \{0, 1\}$ and the post-incidence function $Post: \mathcal{T} \times \mathcal{P} \rightarrow \{0, 1\}$.
Specifically, $Pre(p, t) = 1$ denotes a directed arc from place $p$ to transition $t$, and $Post(t, p) = 1$ denotes a directed arc from $t$ to $p$.
Accordingly, preset $^{\bullet}x$ and postset $x^{\bullet}$ denote the sets of input and output nodes of a node $x \in \mathcal{P} \cup \mathcal{T}$, respectively.
The marking of $\Sigma$ is a mapping $M: \mathcal{P} \rightarrow \mathbb{N}$ that assigns tokens to places, with $M_0$ denoting the initial marking.
A transition $t \in \mathcal{T}$ is enabled at marking $M$ if $\forall p \in \mathcal{P}$, $M(p) \geq Pre(p, t)$, denoted as $M[t\rangle$.
Firing $t$ results in a new marking $M'$ such that $\forall p \in \mathcal{P}, M'(p) = M(p) - Pre(p, t) + Post(t, p)$, denoted as $M[t\rangle M'$.
For a firing sequence $\sigma = (t_1, t_2, \cdots, t_k)$, $M[\sigma\rangle M'$ indicates that the transitions in $\sigma$ fire sequentially, transforming $M$ into $M'$.
$\mathcal{R}(M_0)$ represents the set of markings reachable from $M_0$.
A detailed discussion on PN theory can be found in \cite{david1994petri,li2009deadlock}.

\begin{Definition} \label{definition_petri_net}
	The PN modeling solution $\Pi = (\Pi^+, \Pi^-)$ is defined as $\Sigma(\Pi) = \langle \mathcal{P}, \mathcal{T}, Pre, Post, M_0 \rangle$, where:
	\begin{itemize}
		\item $\mathcal{P} = \mathcal{P}^+ \cup \mathcal{P}^- \cup \mathcal{P}_0^+ \cup \mathcal{P}_0^-$ is the set of places, with $\mathcal{P}^+ = \{ p_i^+ \mid i \in \mathcal{N} \}$, $\mathcal{P}^- = \{ p_i^- \mid i \in \mathcal{N} \}$, $\mathcal{P}_0^+ = \{ p_0^r \mid r \in \mathcal{V}^+ \}$, and $\mathcal{P}_0^- = \{ p_0^r \mid r \in \mathcal{V}^- \}$.
		
		\item $\mathcal{T} = \{ t_i \mid i \in \mathcal{N} \}$ is the set of transitions.
		
		\item $Pre: \mathcal{P} \times \mathcal{T} \rightarrow \{0, 1\}$ defines the input arcs, and for any $i \in \mathcal{N}$:
		\begin{equation} \nonumber
			Pre(p, t_i) = 
			\begin{cases} 
				1, & \text{if } p = p_i^+ \text{ or } p = p_i^- \\
				0, & \text{otherwise}
			\end{cases}
		\end{equation}
		
		\item $Post: \mathcal{T} \times \mathcal{P} \rightarrow \{0, 1\}$ defines the output arcs, and for any $i \in \mathcal{N}$:
		\begin{equation} \nonumber
			Post(t_i, p) = 
			\begin{cases} 
				1, & \text{if } p=p_j^+ \text{ and } (i, j) \in \Pi^+ \\
				1, & \text{if } p=p_k^- \text{ and } (i, k) \in \Pi^- \\
				1, & \text{if } p=p_0^r \text{ and } (i, 0) \in \pi_r , r \in \mathcal{V}\\
				0, & \text{otherwise}
			\end{cases}
		\end{equation}
		
		\item $M_0$ is the initial marking, defined as:
		\begin{equation} \nonumber
			M_0(p) = 
			\begin{cases} 
				1, & \text{if } p = p_i^+ \text{ and } (0, i) \in \Pi^+ \\
				1, & \text{if } p = p_i^- \text{ and } (0, i) \in \Pi^- \\
				0, & \text{otherwise}
			\end{cases}
		\end{equation}
	\end{itemize}
\end{Definition}

To illustrate Definition \ref{definition_petri_net}, Fig. \ref{solution_petri} visualizes the PN model $\Sigma(\Pi)$ constructed from the running example solution $\Pi$.
In this model, each transition $t_i \in \mathcal{T}$ corresponds to a task $i \in \mathcal{N}$.
Governed by the $Pre$ function, each $t_i$ is connected to two distinct input places, $p_i^+$ and $p_i^-$.
This structure strictly enforces synchronization, ensuring that $t_i$ fires only when both the assigned carrier and shuttle are simultaneously available (i.e., both input places hold tokens).
Upon firing, the $Post$ function dictates the downstream flow by mapping the routing precedence in $\Pi$ to directed output arcs.
These arcs guide tokens from $t_i$ to the respective successor places $p_j^+$ (next carrier task) and $p_k^-$ (next shuttle task), thereby enforcing the sequential execution order.
For the terminal task of each route, the token is directed to the virtual depot place $p_0^r$, signifying schedule completion, while $M_0$ initializes the input places of the first tasks to enable immediate execution.

\subsection{Firing-driven Decoding}
As shown in the lower-middle part of Fig.~\ref{fig:framework}, decoding in the ALNS framework maps a static candidate solution to a feasible execution schedule, thereby allowing for the evaluation of solution quality to guide the search process.
Specifically, by simulating the state evolution of the constructed PN model $\Sigma(\Pi)$, the static permutation-based solution $\Pi$ is decoded into a schedule and realized in a dynamic, collaborative execution process.
Based on the spatiotemporal information extracted from this process, the final objective function value is derived.

To formalize the termination of this decoding process, we first introduce the final marking $M_\textnormal{F}$.

\begin{Definition}
	The final marking of PN $\Sigma(\Pi)$, denoted by $M_\textnormal{F}$, is defined such that $\forall p_0^r \in \mathcal{P}_0^+ \cup \mathcal{P}_0^-$, $M_\textnormal{F}(p_0^r) = 1$, $\forall p \in \mathcal{P}^+ \cup \mathcal{P}^-$, $M_\textnormal{F}(p) = 0$.
	\label{terminal_marking}
\end{Definition}

\begin{Theorem}
	A necessary and sufficient condition for PN $\Sigma(\Pi)$ to reach $M_\textnormal{F}$ is that $\forall t \in \mathcal{T}$ is fired exactly once. \label{M_F_reach}
\end{Theorem}

Theorem \ref{M_F_reach} establishes that reaching $M_\textnormal{F}$ is equivalent to the completion of all tasks.
The proofs of all proposed theorems are provided in the supplementary file.
Guided by this criterion, the FDD method is proposed in Algorithm \ref{decoding} to assess solution $\Pi$.
Starting from $M_0$, the algorithm sequentially simulates transition firings to track the system evolution until $M_\textnormal{F}$ is reached.
Theoretically, this process corresponds to generating a trajectory within the reachability graph of $\Sigma(\Pi)$.
Governed by the structure and firing rules of $\Sigma(\Pi)$, this simulation inherently satisfies the synchronization requirements and precedence relations defined in $\Pi$, yielding a spatiotemporal execution profile from which travel distance and tardiness are derived to compute the final objective value.
The computational complexity of FDD is $O(n)$.

The FDD process for the example solution $\Pi$ in Fig. \ref{solution_petri} is represented by the firing sequence $\sigma = (t_5, t_6, t_1, t_4, t_3, t_7, t_2)$, which is detailed in the supplementary file.
It is noted that the firing sequence from $M_0$ to $M_\textnormal{F}$ is not unique.
For example, at $M_0$, transitions $t_5$ and $t_6$ are concurrently enabled.
However, the firing order of concurrent transitions has no impact on the final decoding result.
By Definition~\ref{definition_petri_net}, $\Sigma (\Pi)$ is inherently conflict-free (i.e., $\forall p \in \mathcal{P}, | p^{\bullet} | \leq 1$) due to the deterministic sequential execution of tasks.

\begin{algorithm}[!ht]
	\small
	\SetKwInOut{Input}{input}
	\SetKwInOut{Output}{output}
	\Input{Petri net $\Sigma(\Pi)$, final marking $M_\textnormal{F}$;} 
	\Output{Objective value $F(\Pi$);}
	\BlankLine 
	Initialize current marking $M_\textnormal{C} \leftarrow M_0$;\\
	Initialize enabled transition set $\mathcal{T}_\textnormal{E} \leftarrow \{t_i \mid (0, i) \in \Pi^+\} \cap \{t_i \mid (0, i) \in \Pi^-\}$; \\
	Initialize objective value $F(\Pi) \leftarrow 0$; \\
	
	\While{$M_\textnormal{C} \neq M_\textnormal{F}$}
	{	
		\If{$\mathcal{T}_\textnormal{E} = \emptyset$}
		{	
			$F(\Pi) \leftarrow \infty$; \\
			\textbf{return}{ $F(\Pi)$}. \Comment{\textbf{$\Sigma(\Pi)$ encounters deadlock at $M_\textnormal{C}$}.}  \\
		}
		Randomly select and remove $t_k$ from set $\mathcal{T}_\textnormal{E}$; \\
		Calculate distance $\delta_k$ and tardiness $\tau_k$; \\
		$F(\Pi) \leftarrow F(\Pi) + \lambda \cdot \delta_k + (1 - \lambda) \cdot \tau_k$; \\
		$M_\textnormal{C} [ t_k \rangle M_\textnormal{C}'$ and update $M_\textnormal{C} \leftarrow M_\textnormal{C}'$; \\
		\For{$t \in \bigcup_{p \in t_{k}^{\bullet}} p^{\bullet}$}
		{
			\If{$M_\textnormal{C}[t \rangle$}
			{
				Append $t$ into set $\mathcal{T}_\textnormal{E}$;
			}
		}
	}
	\textbf{return}{ $F(\Pi)$.}
	\caption{Firing-driven decoding}
	\label{decoding} 
\end{algorithm}

\vspace{-10pt}
\subsection{Deadlock-Induced Infeasibility} \label{deadlock_feasibility}

For AHASP whose schedules are tightly coupled by synchronization constraints, a major challenge in applying metaheuristics such as ALNS to AHASP is deadlock-induced infeasibility.
Although each individual route is acyclic, the cascading precedence relations across carriers and shuttles may introduce circular waiting for resources, rendering $\Pi$ infeasible.
In this paper, feasibility refers strictly to structural feasibility. 
Formally, a solution $\Pi$ is termed feasible if and only if its precedence relations are free of circular waits.
As indicated in Line 7 of Algorithm \ref{decoding}, such infeasibility manifests as a deadlock in PN $\Sigma (\Pi)$, causing the state evolution to be blocked before reaching $M_\textnormal{F}$.
The deadlock in PN theory is defined next.
\begin{Definition}
	Deadlock is defined as a marking $M \in \mathcal{R}(M_0)$ such that $\nexists t \in \mathcal{T}$, $M[t\rangle$.
	\label{deadlock}
\end{Definition}

Strictly speaking, the final marking $M_\textnormal{F}$ is a deadlock marking, at which no transition remains enabled.
However, $M_\textnormal{F}$ is a designated target state representing the successful completion of all tasks.
In contrast, the deadlocks addressed in this study refer specifically to structural deadlocks caused by circular wait conditions, where the system halts while tasks remain incomplete.
Fig. \ref{deadlock_example} illustrates a representative example arising from conflicting precedence relations in $\Pi$.
Specifically, task $i$ precedes $j$ in the carrier routes ($(i, j) \in \Pi^+$), implying that the firing of $t_j$ requires a token in $p_j^+$ supplied by $t_i$.
Simultaneously, task $j$ precedes $i$ in the shuttle routes ($(j, i) \in \Pi^-$), implying that the firing of $t_i$ requires a token in $p_i^-$ supplied by $t_j$.
Consequently, a circular wait arises where neither transition can fire (as highlighted by the red loop), preventing the system from reaching $M_\textnormal{F}$.

\begin{Theorem}
	A necessary and sufficient condition for $\Pi$ to be a feasible solution is that PN $\Sigma(\Pi)$ is deadlock-free at $\forall M \in \mathcal{R}(M_0)  \setminus \{M_\textnormal{F}\}$.
	\label{deadlock_feasible}
\end{Theorem}

To assess the feasibility of solution $\Pi$, we establish Theorem \ref{deadlock_feasible} based on deadlock analysis in $\Sigma (\Pi)$.
By combining Theorems \ref{M_F_reach} and \ref{deadlock_feasible}, the absence of deadlock at $\forall M \in \mathcal{R}(M_0)  \setminus \{M_\textnormal{F}\}$ implies the existence of a firing sequence $\sigma$ that contains each transition in $\mathcal{T}$ once, such that $M_0[\sigma \rangle M_\textnormal{F}$.
In physical terms, this means that every task is successfully executed exactly once.
Thus, the feasibility of $\Pi$ can be determined by constructing the reachability graph \cite{li2009deadlock} of PN $\Sigma(\Pi)$, which inherently follows the same process as FDD.
Thus, Algorithm \ref{decoding} unifies decoding and deadlock detection into a single reachability analysis procedure.

\begin{figure}[ht]
	\centering
	\includegraphics[width=0.68\linewidth]{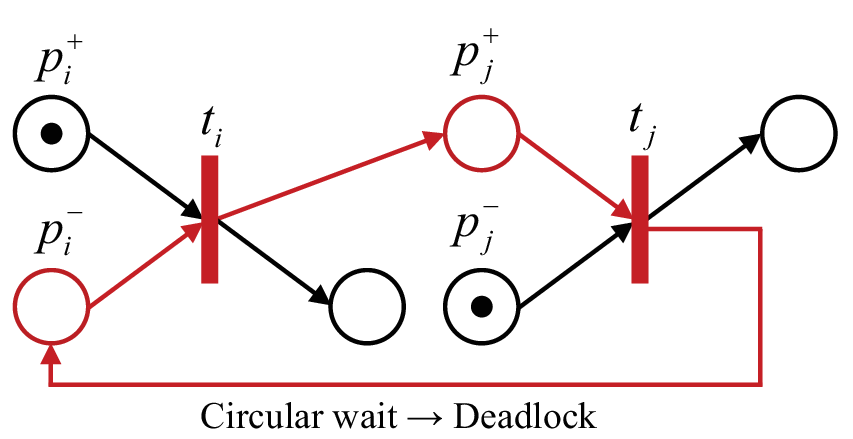}
	\caption{An example of deadlock in PN $\Sigma (\Pi)$ caused by circular wait.}
	\label{deadlock_example}
	\vspace{-10pt}
\end{figure}

\subsection{Deadlock Prevention}\label{position_detection}
Although FDD can detect deadlocks via reachability analysis, relying on post-generation deadlock detection during solution decoding results in a computationally inefficient trial-and-error process that wastes significant computation on infeasible candidates.
To guarantee a deadlock-free search by restricting ALNS insertion operations to strictly feasible positions, we propose BDS to shift from costly a posteriori deadlock detection to efficient a priori prevention.

As illustrated in the “Deadlock-free Construction” module in the upper-middle part of Fig.~\ref{fig:framework}, the core idea of BDS is to leverage the structural analysis of the PN model constructed from the partial solution to preemptively identify insertion positions that would lead to deadlocks.
For the insertion of a removed task $k$, BDS outputs two sets of infeasible positions: the backward set $\mathcal{L_\textnormal{B}}$ and the forward set $\mathcal{L_\textnormal{F}}$.
Specifically, $\mathcal{L_\textnormal{B}}$ contains indices $\overleftarrow{\epsilon}^+(i)$ (resp. $\overleftarrow{\epsilon}^-(i)$), representing the insertion position immediately preceding task $i$ in $\Pi^+$ (resp. $\Pi^-$).
Conversely, $\mathcal{L_\textnormal{F}}$ contains indices $\overrightarrow{\epsilon}^+(i)$ (resp. $\overrightarrow{\epsilon}^-(i)$), representing the insertion position immediately succeeding task $i$.
We assume that the partial solution $\Pi$ prior to insertion is feasible, satisfying the condition in Theorem \ref{deadlock_feasible}.
The BDS method is outlined in Algorithm \ref{feasible_insertion}, which addresses two cases:
(1) \textbf{Unassigned task}:  If task $k$ is not yet assigned to any route (i.e., transition $t_k$ has not been introduced into $\Sigma(\Pi)$), its individual insertion into either $\Pi^+$ or $\Pi^-$ cannot form a circular wait.
Therefore, all candidate positions are feasible at this stage.
(2) \textbf{Partially Assigned Task}: Without loss of generality, assume task $k$ has been inserted into $\Pi^+$ at position $\overrightarrow{\epsilon}^+(i)$ / $\overleftarrow{\epsilon}^+(j)$.
This assignment introduces transition $t_k$, and establishes precedence relations within $\Sigma(\Pi)$, which constrains the insertion of task $k$ into the counterpart $\Pi^-$.
To detect infeasible positions, BDS performs a bidirectional breadth-first search rooted at transition $t_k$ to trace all dependency chains within $\Sigma(\Pi)$.
Specifically, the backward search explores all transitions that possess a directed path to $t_k$, yielding indices $\forall \overleftarrow{\epsilon}^-(b) \in \mathcal{L_\textnormal{B}}$ such that $b \prec k$.
Conversely, the forward search identifies all transitions on directed paths originating from $t_k$, yielding indices $\forall \overrightarrow{\epsilon}^-(f) \in \mathcal{L_\textnormal{F}}$ such that $k \prec f$.
By preemptively excluding these infeasible positions, the resulting "Candidate Solution" in the ALNS framework (Fig. \ref{fig:framework}, right) is guaranteed to be deadlock-free, thereby preventing FDD from wasting resources on invalid candidates.
The soundness and completeness of BDS are proved in the supplementary file.

The computational complexity of BDS is $O(n)$.
Specifically, BDS performs backward and forward breadth-first searches on the dependency graph induced by $\Sigma(\Pi)$.
Breadth-first search has linear complexity in the numbers of searched nodes and directed arcs.
In the induced dependency graph, transitions serve as search nodes, while places encode directed dependencies between transitions.
By Definition~\ref{definition_petri_net}, $|\mathcal{T}|=n$ and $|\mathcal{P}|=2n+m^+ + m^-$.
Moreover, since each place has at most one input transition and one output transition, the number of directed dependency arcs induced by the places is also linear in $|\mathcal{P}|$.
Therefore, the numbers of searched nodes and directed arcs are both linear in $n$, yielding a single-run complexity of $O(n)$ for BDS.

An illustrative example of BDS is presented in Fig. \ref{BDS_example}, where task 8 is inserted at position $\overrightarrow{\epsilon}^+(1)$ / $\overleftarrow{\epsilon}^+(3)$ in the solution $\Pi$ shown in Fig. \ref{solution_petri}.
The corresponding modification in PN $\Sigma (\Pi)$ is highlighted in the red box.
The backward and forward searches are depicted using arrows of two different colors in $\Sigma (\Pi)$.
Note that the direction of the backward search is opposite to that of the blue arrows.
The backward search identifies three infeasible positions: $\mathcal{L_\textnormal{B}} = \{\overleftarrow{\epsilon}^-(1), \overleftarrow{\epsilon}^-(5), \overleftarrow{\epsilon}^-(6)\}$.
The forward search detects three other invalid positions: $\mathcal{L_\textnormal{F}} = \{\overrightarrow{\epsilon}^-(2), \overrightarrow{\epsilon}^-(3), \overrightarrow{\epsilon}^-(7)\}$.
Only three feasible insertion positions exist: $\overleftarrow{\epsilon}^-(4)$, $\overleftarrow{\epsilon}^-(7)$, and $\overrightarrow{\epsilon}^-(6)$, reflecting an appreciable proportion of infeasible positions.

\begin{figure}[tb]
	\centering
	\includegraphics[width=0.97\linewidth]{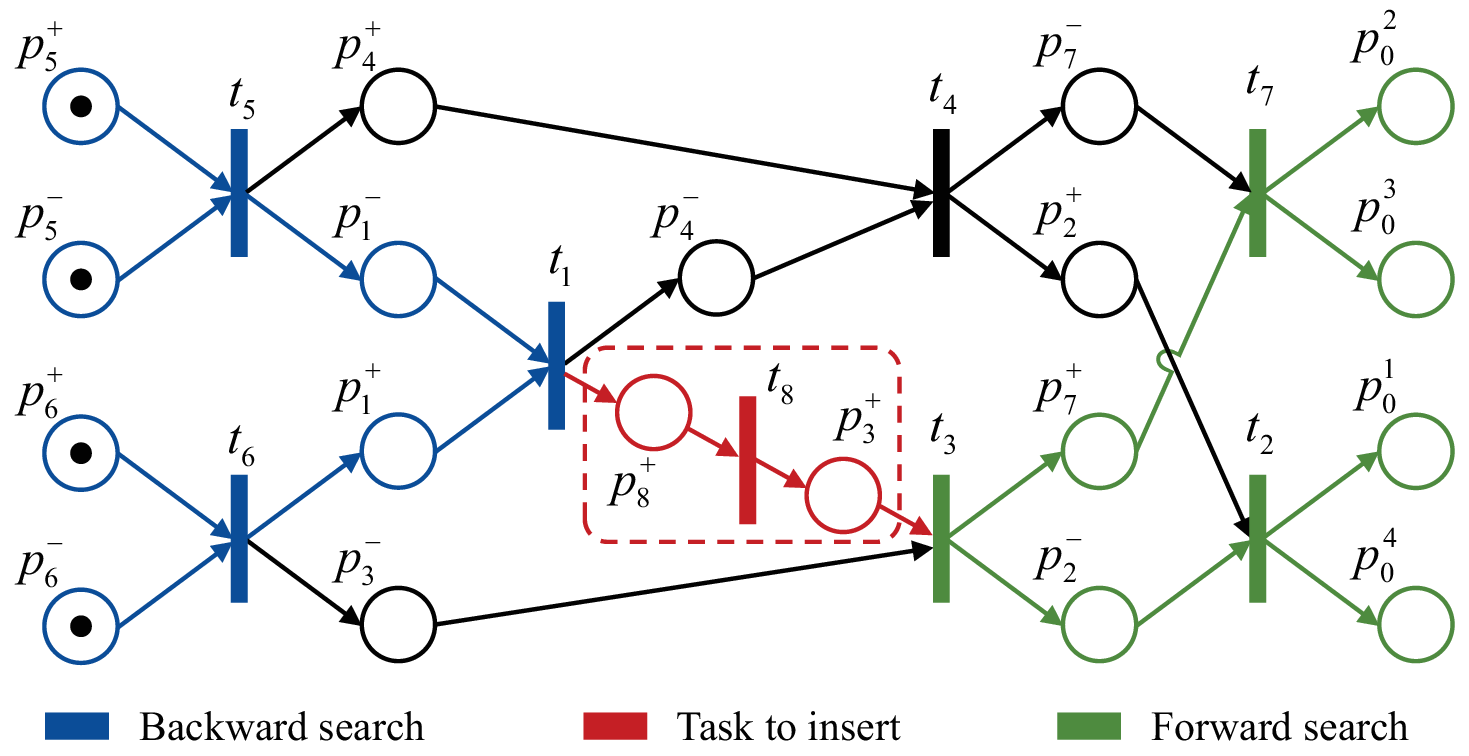}
	\caption{BDS in PN $\Sigma(\Pi)$ with task 8 inserted at $\overrightarrow{\epsilon}^+(1)$ / $\overleftarrow{\epsilon}^+(3)$.}
	\label{BDS_example}
	\vspace{-10pt}
\end{figure}

\begin{algorithm}[!ht]
	\small
	\SetKwInOut{Input}{input}
	\SetKwInOut{Output}{output}
	\Input{PN $\Sigma(\Pi)$, task $k$ to be inserted;} 
	\Output{Infeasible insertion positions for task $k$ in $\Pi^-$;}
	\BlankLine
	Initialize infeasible position set $\mathcal{L_\textnormal{B}} \leftarrow \emptyset$ and $\mathcal{L_\textnormal{F}} \leftarrow \emptyset$; \\
	\If{$t_k \notin \mathcal{T}$}
	{
		\Return{$\emptyset$.} 
	}
	Initialize backward and forward queues $\mathcal{Q}_{\textnormal{B}} \leftarrow \cup_{p \in {}^{\bullet}t_{k}} {}^{\bullet}p$, $ \quad \mathcal{Q}_{\textnormal{F}} \leftarrow \cup_{p \in t_{k}^\bullet} p^\bullet$;\\
	\While{$\mathcal{Q}_{\textnormal{B}} \neq \emptyset$}
	{	
		$t_b \gets \textnormal{\textsc{Dequeue}}(\mathcal{Q}_{\textnormal{B}})$; \\
		Append position $\overleftarrow{\epsilon}^-(b)$ into set $\mathcal{L_\textnormal{B}}$; \\
		$\mathcal{Q_\textnormal{B}} \leftarrow \left[ \left(\cup_{p \in {}^{\bullet}t_{b}} {}^{\bullet}p\right) \setminus \left\{t_d | \overleftarrow{\epsilon}^-(d) \in \mathcal{L_\textnormal{B}} \right\} \right] \bigcup \mathcal{Q_\textnormal{B}}$; \\
	}
	$\triangleright$ \textbf{Forward search}\\
	\While{$\mathcal{Q}_{\textnormal{F}} \neq \emptyset$}
	{
		$t_f \gets \textnormal{\textsc{Dequeue}}(\mathcal{Q}_{\textnormal{F}})$; \\
		Append position $\overrightarrow{\epsilon}^-(f)$ into set $\mathcal{L_\textnormal{F}}$; \\
		$\mathcal{Q_\textnormal{F}} \leftarrow \left[ \left( \cup_{p \in t_{f}^{\bullet}} p^{\bullet}\right) \setminus \{t_d | \overrightarrow{\epsilon}^-(d) \in \mathcal{L_\textnormal{F}}\} \right] \bigcup \mathcal{Q_\textnormal{F}}$; \\
	}
	
	\textbf{return}{ $\mathcal{L_\textnormal{B}} \cup \mathcal{L_\textnormal{F}}$.}
	\caption{Bidirectional Dependency Search}
	\label{feasible_insertion} 
\end{algorithm}

\section{Petri Net-enhanced metaheuristic} \label{Adaptive Large Neighborhood Search}

Having established the PN modeling and analysis strategies in Section \ref{Petri_feasibility} (Fig. \ref{fig:framework}, left and middle panels), we now proceed to construct the upper-level optimization framework.
As illustrated in the right panel of Fig. \ref{fig:framework}, this section details the PN-enhanced ALNS algorithm, specifically describing its architecture, operators, and the computational acceleration achieved via BDS-based deadlock prevention.

\subsection{Neighborhood Operators}\label{Neighborhood operators}
Inspired by \cite{ropke2006adaptive,voigt2024review}, we develop five removal operators and three insertion operators, collectively denoted as $\mathcal{O^\textnormal{R}}$ and $\mathcal{O^\textnormal{I}}$, respectively.
In each iteration of ALNS, a removal operator is selected from $\mathcal{O^\textnormal{R}}$ to remove $\phi \cdot n$ tasks from the current solution $\Pi_\textnormal{C}$, placing them into a removed task set $\mathcal{N^\textnormal{R}}$.
Parameter $\phi \in [0,1]$ denotes the proportion of tasks to be removed relative to the total number of tasks.
The resulting partially destroyed solution is denoted as $\Pi_\textnormal{D}$.
An insertion operator is then selected from $\mathcal{O^\textnormal{I}}$ to reinsert the tasks from $\mathcal{N^\textnormal{R}}$ into $\Pi_\textnormal{D}$, yielding a new solution $\Pi_\textnormal{N}$.

The operators are selected by using a roulette wheel principle based on their historical performance.
Each operator $o \in \mathcal{O^\textnormal{R}} \cup \mathcal{O^\textnormal{I}}$ is assigned a weight $\omega_o$, initialized to 1.
The selection probability for a removal operator $o^-$ is calculated as $\omega_{o^-} / \sum_{o \in \mathcal{O^\textnormal{R}} }\omega_o$, with the same principle applied to insertion operators.
After $\kappa$ iterations, the weight of operator $o$ is updated to $\omega_{o} = (1-\rho)\cdot\omega_{o} + \rho \cdot \xi_{o} / \theta_{o}$, where $\rho \in [0,1]$ denotes the reaction factor, and $\xi_{o}$ and $\theta_{o}$ represent the accumulated score and the usage count of operator $o$, respectively.
Details of the scoring rules are provided in \cite{ropke2006adaptive}.

\subsubsection{Removal Operators}
We employ five distinct operators to remove $\phi \cdot n$ tasks from the current solution $\Pi_\textnormal{C}$.
\textit{Random Removal} removes tasks stochastically to diversify the search.
The greedy heuristics include \textit{Longest Distance Removal}, \textit{Largest Tardiness Removal}, and \textit{Highest Cost Removal}, which remove tasks based on the longest travel distance $\delta_i$, the greatest tardiness $\tau_i$, and the highest weighted cost $F_i = \lambda \cdot \delta_i + (1 - \lambda) \cdot \tau_i$, respectively.
Additionally, \textit{Shaw Removal} \cite{voigt2024review} removes a reference task and its most related neighbors based on a spatiotemporal difference metric $\Omega(i,j) = \psi \cdot (d_{s_i, s_j} + d_{e_i, e_j}) + (1 - \psi) \cdot |T_i^\textnormal{c} - T_j^\textnormal{c}| $.

\subsubsection{Insertion Operators}
To reinsert tasks from the removed set $\mathcal{N^\textnormal{R}}$, we employ three operators.
While all operators greedily place the selected task into the position that minimizes the objective value increase, they differ in the task selection order.
\emph{Greedy Insertion} selects the next task from $\mathcal{N^\textnormal{R}}$ randomly.
In contrast, \emph{Urgency-prioritized} and \emph{Cost-prioritized Greedy Insertion} prioritize tasks based on the earliest deadline ($i = \mathop{\arg\min}_{j \in \mathcal{N^\textnormal{R}}} T_j^\textnormal{c}$) and the maximum cost ($i = \mathop{\arg\max}_{j \in \mathcal{N^\textnormal{R}}} F_j$), respectively.

\subsection{BDS-Accelerated Deadlock-Free Greedy Insertion}\label{BDS-based_acceleration}

This section theoretically analyzes the computational acceleration achieved by integrating the BDS-based deadlock-free construction into the ALNS framework.
By preemptively pruning infeasible positions, BDS ensures that FDD is employed exclusively to evaluate feasible neighborhoods during greedy insertion, thereby eliminating redundant evaluations and enhancing computational efficiency.
This acceleration effect is quantified theoretically below.

Given a destroyed solution $\Pi_\textnormal{D}$ with $n$ tasks, $m^+$ carriers, and $m^-$ shuttles, the total number of pairwise position combinations is $N_{\textnormal{total}} = (n + m^+)\cdot(n + m^-)$.
Let $T_{\textnormal{FDD}}$ and $T_{\textnormal{BDS}}$ denote the average computational cost (i.e., runtime) of a single FDD evaluation and a single BDS operation, respectively.
Without BDS, the standard greedy insertion assesses every position combination via FDD, yielding a total computational cost of:
\begin{gather}
	C_\textnormal{O} = T_{\textnormal{FDD}} \cdot(n+m^+)\cdot(n+m^-) \label{complxity_all}
\end{gather}
With the introduction of BDS, for each candidate position $\epsilon \in \mathcal{L}^+$ (where $\mathcal{L}^+$ is the set of carrier positions with $|\mathcal{L}^+| = n + m^+$), the corresponding set of infeasible shuttle positions $\mathcal{L_\textnormal{I}^-}(\epsilon)$ is proactively eliminated via a single BDS calculation.
FDD is then executed only for the remaining feasible combinations.
Consequently, the computational cost of the insertion procedure with BDS becomes:
\begin{gather}
	C_\textnormal{A} = T_{\textnormal{FDD}} \cdot \left[\sum_{\epsilon \in \mathcal{L^+}} \left (n + m^- - |\mathcal{L_\textnormal{I}^-}(\epsilon)| \right)\right] + T_{\textnormal{BDS}} \cdot (n + m^+) \label{complexity_BDS}
\end{gather}
Given that FDD and BDS share a linear complexity of $O(n)$ and exhibit comparable empirical runtimes, we approximate $T_{\textnormal{BDS}} \approx T_{\textnormal{FDD}}$.
The acceleration ratio $R_\textnormal{C}$ of $C_\textnormal{A}$ to $C_\textnormal{O}$ is thus derived as:
\begin{gather}
	R_\textnormal{C} =\frac{C_\textnormal{A}}{C_\textnormal{O}}  = 1 - \frac{\sum_{\epsilon \in \mathcal{L^+} }|\mathcal{L_\textnormal{I}^-}(\epsilon)| - (n + m^+)}{(n + m^+)\cdot(n + m^-)} \label{ratio}
\end{gather}
which indicates that the acceleration efficiency depends on the proportion of infeasible positions pruned.
A threshold case arises when $|\mathcal{L_\textnormal{I}^-}(\epsilon)| = 1$ on average, resulting in $R_\textnormal{C} = 1$, which implies the overhead of BDS balances out the pruning benefits.
However, in the context of AHASP, the tight interdependence among AGVs typically results in a high probability of potential deadlocks, meaning $|\mathcal{L_\textnormal{I}^-}(\epsilon)|$ is significantly larger than one.
Consequently, the subtraction term in \eqref{ratio} becomes substantial, driving $R_\textnormal{C}$ well below 1.
This result theoretically validates the computational acceleration of the proposed PN-based deadlock-free framework for the ALNS algorithm.

\subsection{Framework of ALNS}
The framework of ALNS is shown in the right panel of Fig. \ref{fig:framework}.
ALNS starts by constructing an initial solution through greedy insertion of all tasks.
In each iteration, a pair of removal and insertion operators is selected to generate a new solution $\Pi_\textnormal{N}$.
Subsequently, $\Pi_\textnormal{N}$ is evaluated by using FDD.
Then, a simulated annealing-like acceptance criterion determines whether to accept $\Pi_\textnormal{N}$ with a fixed temperature $T_\textnormal{F} = \mu \cdot \sum_{i \in \mathcal{N}}d_{s_i, e_i} / [n \cdot (m^+ + m^-)]$ \cite{ruiz2019iterated}.
Here, $\mu$ is the temperature coefficient requiring calibration.
Following this, the weights of the selected operators are updated.
Finally, when the termination condition is met, the best solution is returned.
To ensure practical responsiveness while adapting to varying problem scales, we adopt a time-based stopping criterion.
The maximum CPU runtime (in seconds) is set according to instance complexity as $\Delta T = \zeta \cdot n^2 \cdot (m^+ + m^-)$, where $n$, $m^+$, and $m^-$ denote the number of tasks, carriers, and shuttles, respectively, and $\zeta$ is a scaling factor determined by the required responsiveness.

\section{Numerical Experiment}\label{Numerical Experiment}

This section conducts extensive experiments to validate the computational efficiency of the \textit{PN-based deadlock-free scheduling framework} and the optimization performance of the developed \textit{ALNS}.
First, we quantify the acceleration effects and solution improvements enabled by the proactive deadlock prevention mechanism.
Next, we benchmark the proposed ALNS against exact solvers and state-of-the-art metaheuristics to demonstrate its competitiveness. 
Subsequently, we conduct operational analyses and sensitivity tests to derive managerial insights. 
Finally, the generalizability of the proposed framework is verified on general HRSP instances.

\subsection{Experimental Setting}
Given the specific industrial context of AHASP, we collected a comprehensive dataset from the battery manufacturing facility shown in Fig. \ref{factory}.
To safeguard data privacy, the dataset is aggregated and perturbed while preserving practical relevance.
The dataset is categorized into five instance scales based on task count: $n=5, 10, 20, 30, \textnormal{and } 40$, where 40 tasks correspond to the largest simultaneously scheduled workload observed in the studied real-world manufacturing system.
Each group contains 10 instances, with a total of 50.
Instances are labeled in the format T10\_I1, denoting the first instance with 10 tasks.
This manufacturing system is equipped with four carriers and eight shuttles ($m^+$ = 4 and $m^-$ = 8).
The AGVs operate at a velocity of $v$ = 1.2m/s.
The attachment, detachment, and load durations are set to $t^\textnormal{a}$ = 8s, $t^\textnormal{d}$ = 8s, and $t^\textnormal{p}$ = 30s, respectively.
To balance the distance and tardiness, $\lambda$ is set to 0.4.
Parameter $\zeta$ is set to 0.01 according to practical response requirements.

Algorithm performance is assessed using the relative percentage deviation (RPD), defined as: 
\begin{gather}
	\textnormal{R}_{\textnormal{PD}} =  \frac{F-F^\dagger}{F^\dagger} \times 100\% \label{metric}
\end{gather}
where $F$ denotes the minimum objective value obtained by a specific algorithm for a given instance, and $F^\dagger$ denotes the best (minimal) objective value achieved by all algorithms for that instance.

All algorithms are programmed in Python and run on a computer with AMD Ryzen 9 7945HX 2.5GHz CPU, 32GB of RAM, and 64-bit Windows 11.

\subsection{Parameter Tuning} \label{Parameter Tuning}
The proposed metaheuristic involves five parameters: temperature coefficient $\mu$, removal proportion $\phi$, reaction factor $\rho$, update interval $\kappa$, and Shaw removal weight $\psi$.
To identify the optimal parameter configuration, we apply the Taguchi method for the design of experiments \cite{montgomery2017design} by using instance T30\_I1.
Each of the five parameters is tested at four levels, forming an orthogonal array $L_{16}(4^5)$.
Each parameter configuration is independently run 20 times, and the main effects of the five parameters are demonstrated in Fig. \ref{main_effect}.
As the Smaller-the-Better criterion is adopted, smaller AOVs (average objective values) in Fig. \ref{main_effect} imply better parameter configurations.
Thus, the parameters are configured as $\mu$ = 20, $\phi$ = 0.3, $\rho$ = 0.1, $\kappa$ = 15, and $\psi$ =0.3.
Given that ALNS is relatively insensitive to parameter settings due to its self-adaptive mechanism \cite{ropke2006adaptive}, this configuration is fixed for all subsequent experiments to verify robustness across varying operational conditions.

\begin{figure}[ht]
	\centering
	\includegraphics[width=0.92\linewidth]{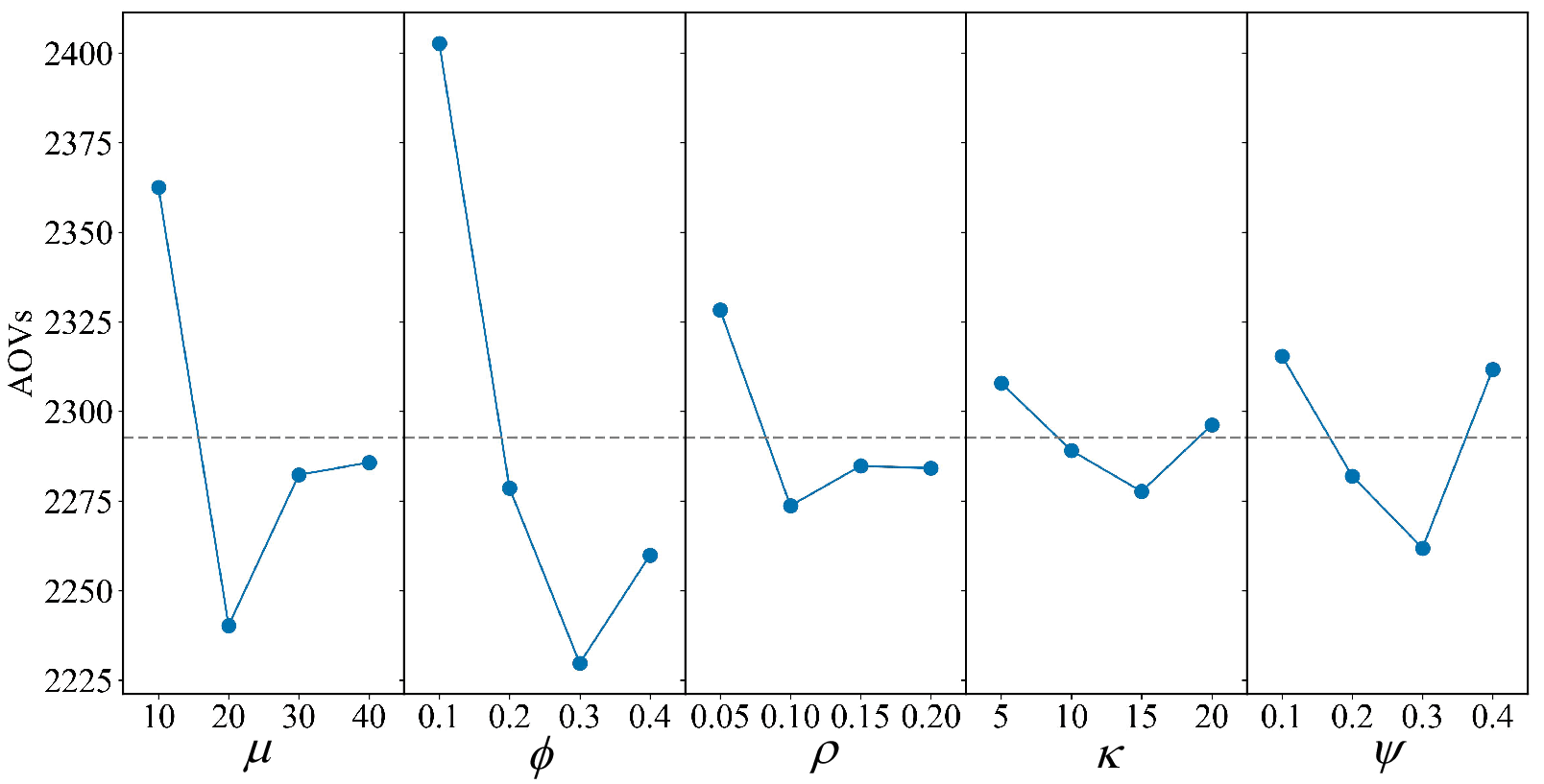}
	\caption{Main effects plot for the five key parameters of ALNS.}
	\label{main_effect}
	\vspace{-10pt}
\end{figure}

\subsection{Evaluation of PN-based deadlock-free framework}
This section evaluates the effectiveness of the proposed PN-based deadlock-free scheduling framework.
By integrating BDS, we shift from a posteriori deadlock detection based on FDD to a priori deadlock prevention enabled by PN structural analysis.
To quantify the resulting computational acceleration and optimization gains, we benchmark the BDS-enhanced ALNS with active deadlock prevention against a baseline ALNS that relies exclusively on FDD for passive deadlock detection.
Under the fixed time limit $\Delta_T$ defined in Section V-C, we record the iteration counts ($\chi^\textnormal{A}$ and $\chi^\textnormal{O}$) and RPD values ($\textnormal{R}_{\textnormal{PD}}^\textnormal{A}$ and $\textnormal{R}_{\textnormal{PD}}^\textnormal{O}$) achieved in the two settings under varying numbers of tasks and AGVs.
Three metrics, namely the acceleration ratio ($\textnormal{A}_\textnormal{R}$), the gap between RPDs ($\textnormal{G}_\textnormal{AP}$), and the infeasible position proportion ($\textnormal{I}_\textnormal{PP}$), are defined as follows:
\begin{gather}
	\textnormal{A}_\textnormal{R} = \frac{\chi^\textnormal{A} - \chi^\textnormal{O}}{\chi^\textnormal{O}} \times 100\%  \label{AR}\\ 
	\textnormal{G}_\textnormal{AP} = \textnormal{R}_{\textnormal{PD}}^\textnormal{O} - \textnormal{R}_{\textnormal{PD}}^\textnormal{A} \label{GAP} \\
	\textnormal{I}_\textnormal{PP} = \frac{\sum_{\epsilon \in \mathcal{L^+} }|\mathcal{L_\textnormal{I}^-}(\epsilon)|}{(n + m^+) \cdot (n + m^-)}  \times 100\% \label{IPP}
\end{gather}
The average CPU runtime per execution for FDD and BDS is also recorded.
Each instance is run 20 times, and the average metric values are reported in Table \ref{BDS_acceleration_method}.
The notation C$x$\_S$y$ represents a fleet configuration comprising $x$ carriers and $y$ shuttles.

As shown in Table~\ref{BDS_acceleration_method}, the proposed framework exhibits strong computational scalability, with the runtimes of both FDD and BDS increasing approximately linearly with the number of tasks.
The active deadlock prevention framework becomes particularly effective in highly constrained, large-scale scenarios.
Specifically, for the tightly coupled C2\_S4 instances ($n=40$), the acceleration ratio ($\textnormal{A}_\textnormal{R}$) reaches 154.56\%, reflecting a substantial increase in search iterations by proactively filtering out a large proportion of infeasible candidates ($\textnormal{I}_\textnormal{PP} = 69.58\%$).
Even in loosely coupled settings with sufficient fleet resources (e.g., C5\_S10), the framework maintains a robust acceleration of over 40\% for large-scale instances with 40 tasks.
Moreover, this accelerated exploration contributes to improved solution quality.
As indicated by the positive $\textnormal{G}_\textnormal{AP}$, the computational resources saved from avoiding invalid evaluations are reallocated to more effective exploration of the feasible solution space, leading to better optimization performance.

\begin{table*}[ht]
	\footnotesize
	\caption{Average Improvement ratios in iterations and objectives with PN-based deadlock-free scheduling framework.}
	\centering
	\begin{tabular}{ccccccccccccccc}
		\hline
		\multirow{2}{*}{\makecell{Task\\Number} } & \multicolumn{2}{c}{CPU Runtime ($\mu s$)} & \multicolumn{3}{c}{C2\_S4} & \multicolumn{3}{c}{C3\_S6} & \multicolumn{3}{c}{C4\_S8} & \multicolumn{3}{c}{C5\_S10} \\ \cmidrule(r){2-3} \cmidrule(r){4-6} \cmidrule(r){7-9} \cmidrule(r){10-12} \cmidrule(r){13-15}
		& FDD    & BDS   & $\textnormal{I}_\textnormal{PP}$ & $\textnormal{A}_\textnormal{R}$   & $\textnormal{G}_\textnormal{AP}$& $\textnormal{I}_\textnormal{PP}$ & $\textnormal{A}_\textnormal{R}$  & $\textnormal{G}_\textnormal{AP}$& $\textnormal{I}_\textnormal{PP}$ & $\textnormal{A}_\textnormal{R}$  & $\textnormal{G}_\textnormal{AP}$& $\textnormal{I}_\textnormal{PP}$ & $\textnormal{A}_\textnormal{R}$  & $\textnormal{G}_\textnormal{AP}$ \\ \hline
		5    & 14.60  & 9.64  & 26.43 & 25.09  & 0.27 & 14.64 & 22.21 & 0.41 & 11.16 & 19.75 & 0.11 & 8.07   & 14.54 & 0.13 \\
		10   & 28.72  & 18.63 & 38.80 & 49.44  & 0.78 & 22.03 & 34.26 & 0.28 & 14.65 & 24.62 & 0.74 & 11.24  & 21.95 & 0.16 \\
		20   & 57.74  & 37.47 & 55.95 & 84.37  & 2.29 & 33.12 & 54.46 & 1.43 & 21.23 & 37.25 & 2.65 & 15.03  & 30.52 & 1.31\\
		30   & 87.37  & 59.20 & 63.99 & 120.94 & 3.08 & 39.45 & 70.18 & 4.44 & 25.84 & 50.15 & 4.06 & 17.82  & 37.46 & 2.07 \\
		40   & 123.12 & 84.84 & 69.58 &154.56  & 2.34 & 43.40 & 87.19 & 3.93 & 29.82 & 62.21 & 3.81 & 20.07 & 44.93 & 2.86 \\ \hline
	\end{tabular}%
	\label{BDS_acceleration_method}
	\vspace{-5pt}
\end{table*}

\subsection{Comparison Experiments for ALNS}
Given the novelty of AHASP and the lack of directly comparable benchmarks in the literature, we compare the developed ALNS with four state-of-the-art metaheuristics for related AGV scheduling problems: variable neighborhood search (VNS) \cite{wang2024effective}, iterated greedy (IG) algorithm \cite{zou2023effective}, genetic algorithm (GA) \cite{liu2023improved}, and discrete artificial bee colony (DABC) algorithm \cite{zou2020effective}.
VNS in 2024 and IG in 2023 share a similar framework to ALNS as individual-based algorithms, differing mainly in neighborhood strategies.
GA in 2023 and DABC in 2020 are two representative population-based algorithms.
To guarantee a fair comparison, all comparative metaheuristics incorporate the proposed PN-based deadlock-free framework, and the termination criterion is a unified CPU time limit of $\Delta T$.
Additionally, we employ Gurobi (version 11.0) with a deterministic Branch-and-Cut algorithm for the MILP model, and the Google OR-Tools CP-SAT solver for constraint programming, serving as exact solution benchmarks.
Moreover, the Earliest Due Date Bidding-based (EDDBID) scheduling policy commonly used in practice is included for comparison \cite{singh2022matheuristic}.
Each metaheuristic is independently run 20 times for each instance.
The average (Avg) and standard deviation (Std) of the obtained RPDs for instances with 5 and 10 tasks are presented in Tables~\ref{result_5} and \ref{result_10}, with the best Avg highlighted in bold.
For Gurobi, we report RPDs under $\Delta T$ and MIP gaps under extended time limits (60s for 5 tasks and 3600s for 10 tasks).
For OR-Tools, we report the CPU time to optimality for 5 tasks, and the resulting RPDs and optimality gaps under a 3600s limit for 10 tasks.

As shown in Table \ref{result_5}, EDDBID performs well on only a few instances with 5 tasks, such as T5\_I5 and T5\_I9, where RPDs are below 3\%.
Despite its short computation time (averaging 90$\mu$s), EDDBID performs poorly and is unstable on most small-scale instances.
Regarding exact methods, Gurobi struggles to certify global optimality, maintaining significant MIP gaps (avg. 21.27\%) even with extended runtime.
In contrast, OR-Tools efficiently proves optimality for all instances within 1.34s (avg. 0.82s), highlighting the advantage of constraint propagation in handling tight synchronization logic.
Benchmarked against these exact solutions, the metaheuristics exhibit competitive performance, maintaining mean RPDs below 2\% and mean Stds below 1\%.
Notably, the developed ALNS matches the proven optimal solutions, consistently obtaining the best solution in every run across all 5-task instances, outperforming other comparative heuristics in stability.

According to Table \ref{result_10}, the performance of EDDBID deteriorates significantly when $n=10$, with a mean RPD exceeding 60\%.
The scalability limitations of exact methods become pronounced at this scale.
Under the strict time limit $\Delta T = 12$s, Gurobi fails to find a feasible solution for instance T10\_I6, while OR-Tools yields a high average RPD of 67.71\%.
Even with runtime extended to 3600s, although OR-Tools identifies the best-known solutions, it fails to certify global optimality, retaining an average gap of 6.24\%.
Gurobi struggles even more, maintaining a substantial optimality gap of 50.67\%.
Thus, the prohibitive computational cost renders exact methods unsuitable for real-time applications.
For heuristic algorithms, GA and VNS exhibit a notable increase in RPD, while the mean RPDs of IG and DABC surpass 3\%.
In contrast, ALNS demonstrates superior efficiency, attaining a near-optimal mean RPD of 1.34\% within just 12s, comparable to the solution quality obtained by OR-Tools in 1 hour.

\begin{table*}[ht]
	\footnotesize
	\caption{Experimental results for instances with 5 tasks ($\Delta T = 3$s).}
	\label{result_5}
	\centering
	\begin{tabular}{cccccccccccccccc}
		\toprule
		&  & \multicolumn{3}{c}{Gurobi} & OR-Tools & \multicolumn{2}{c}{VNS} & \multicolumn{2}{c}{IG} & \multicolumn{2}{c}{GA} & \multicolumn{2}{c}{DABC} & \multicolumn{2}{c}{Ours} \\ 
		
		\cmidrule(lr){3-5} \cmidrule(lr){6-6} \cmidrule(lr){7-8} \cmidrule(lr){9-10} \cmidrule(lr){11-12} \cmidrule(lr){13-14} \cmidrule(lr){15-16}
		
		\multirow{-2}{*}{Instance} & \multirow{-2}{*}{EDDBID} & RPD (3s) & RPD (60s) & Gap & Time (s) & Avg & Std & Avg & Std & Avg & Std & Avg & Std & Avg & Std \\ \hline
		
		T5\_I1 & 49.29 & 19.25 & \textbf{0.00} & 0.00 & 1.01 & 3.29 & 6.74 & \textbf{0.00} & 0.00 & 1.43 & 1.47 & \textbf{0.00} & 0.00 & \textbf{0.00} & 0.00 \\
		T5\_I2 & 7.76 & 29.63 & \textbf{0.00} & 13.54 & 0.90 & \textbf{0.00} & 0.00 & 1.97 & 1.66 & 3.37 & 0.13 & \textbf{0.00} & 0.00 & \textbf{0.00} & 0.00 \\
		T5\_I3 & 20.24 & \textbf{0.00} & \textbf{0.00} & 72.06 & 0.83 & \textbf{0.00} & 0.00 & \textbf{0.00} & 0.00 & \textbf{0.00} & 0.00 & \textbf{0.00} & 0.00 & \textbf{0.00} & 0.00 \\
		T5\_I4 & 47.25 & 134.34 & \textbf{0.00} & 0.00 & 0.66 & 0.39 & 1.76 & \textbf{0.00} & 0.00 & 1.16 & 1.63 & \textbf{0.00} & 0.00 & \textbf{0.00} & 0.00 \\
		T5\_I5 & 2.33 & 52.09 & \textbf{0.00} & 0.00 & 0.63 & \textbf{0.00} & 0.00 & \textbf{0.00} & 0.00 & \textbf{0.00} & 0.00 & \textbf{0.00} & 0.00 & \textbf{0.00} & 0.00 \\
		T5\_I6 & 5.88 & 36.63 & \textbf{0.00} & 57.93 & 1.34 & \textbf{0.00} & 0.00 & 1.39 & 2.51 & 6.01 & 0.00 & \textbf{0.00} & 0.00 & \textbf{0.00} & 0.00 \\
		T5\_I7 & 97.47 & 4.81 & \textbf{0.00} & 0.00 & 0.86 & \textbf{0.00} & 0.00 & \textbf{0.00} & 0.00 & 0.51 & 1.56 & \textbf{0.00} & 0.00 & \textbf{0.00} & 0.00 \\
		T5\_I8 & 20.37 & 21.50 & \textbf{0.00} & 69.17 & 0.67 & \textbf{0.00} & 0.00 & \textbf{0.00} & 0.00 & 0.20 & 0.61 & \textbf{0.00} & 0.00 & \textbf{0.00} & 0.00 \\
		T5\_I9 & 2.34 & 44.88 & \textbf{0.00} & 0.00 & 0.66 & 0.04 & 0.17 & 0.63 & 1.13 & 0.43 & 0.40 & 0.04 & 0.17 & \textbf{0.00} & 0.00 \\
		T5\_I10 & 59.55 & 46.98 & \textbf{0.00} & 0.00 & 0.64 & 0.28 & 1.26 & \textbf{0.00} & 0.00 & 2.70 & 3.98 & 0.41 & 1.56 & \textbf{0.00} & 0.00 \\ \hline
		Mean & 31.25 & 39.01 & 0.00 & 21.27 & 0.82 & 0.40 & 0.99 & 0.40 & 0.53 & 1.58 & 0.98 & 0.04 & 0.17 & 0.00 & 0.00 \\ \bottomrule
	\end{tabular}%
	\vspace{-10pt}
\end{table*}

\begin{table*}[ht]
	\footnotesize
	\caption{Experimental results for instances with 10 tasks ($\Delta T = 12$s).}
	\label{result_10}
	\centering
	\begin{threeparttable}
		\setlength{\tabcolsep}{3.5pt} 
		
		\resizebox{\textwidth}{!}{%
			\begin{tabular}{cccccccccccccccccc}
				\toprule
				&  & \multicolumn{3}{c}{Gurobi} & \multicolumn{3}{c}{OR-Tools} & \multicolumn{2}{c}{VNS} & \multicolumn{2}{c}{IG} & \multicolumn{2}{c}{GA} & \multicolumn{2}{c}{DABC} & \multicolumn{2}{c}{Ours} \\ 
				
				\cmidrule(lr){3-5} \cmidrule(lr){6-8} \cmidrule(lr){9-10} \cmidrule(lr){11-12} \cmidrule(lr){13-14} \cmidrule(lr){15-16} \cmidrule(lr){17-18}
				
				\multirow{-2}{*}{Instance} & \multirow{-2}{*}{EDDBID} & RPD (12s) & RPD (3600s) & Gap & RPD (12s) & RPD (3600s) & Gap & Avg & Std & Avg & Std & Avg & Std & Avg & Std & Avg & Std \\ \midrule
				
				T10\_I1 & 89.85 & 882.77 & 12.20 & 91.52 & 72.60 & \textbf{0.00} & 18.57 & 6.68 & 2.41 & 4.85 & 1.99 & 9.01 & 1.13 & 3.08 & 1.92 & \textbf{3.04} & 1.42 \\
				T10\_I2 & 69.61 & 343.51 & 12.97 & 20.87 & 80.12 & \textbf{0.00} & 7.35 & 5.33 & 5.21 & 2.79 & 2.86 & 8.80 & 4.49 & 5.55 & 2.25 & \textbf{0.52} & 0.84 \\
				T10\_I3 & 56.81 & 325.59 & 7.79 &  91.01 & 15.24 & \textbf{0.00} & 4.02 & 4.28 & 1.00 & 3.71 & 1.51 & 2.92 & 1.09 & 2.91 & 1.12 & \textbf{2.20} & 0.96 \\
				T10\_I4 & 27.37 & 819.79 & 9.01 & 11.45 &  80.72 & \textbf{0.00} &  3.82 & 4.11 & 4.59 & 2.33 & 2.33 & 10.93 & 2.63 & 4.85 & 2.24 & \textbf{0.79} & 1.21 \\
				T10\_I5 & 95.81 & 523.29 & 15.98 & 28.29 &  78.81 &  \textbf{0.00} & 12.69 & 15.68 & 4.27 & 6.45 & 3.17 & 9.76 & 1.99 & 7.63 & 0.69 & \textbf{3.82} & 2.30 \\
				T10\_I6 & 59.81 & -- & 23.82 &  81.65 &  79.19 & \textbf{0.00} &  2.08 & 7.74 & 7.53 & 2.69 & 2.53 & 8.93 & 3.20 & 1.54 & 0.68 & \textbf{1.40} & 1.03 \\
				T10\_I7 & 74.15 & 896.44 & \textbf{0.00} &  4.36 &  73.46 & \textbf{0.00} &  0.00 & 1.00 & 2.15 & 0.29 & 0.37 & 0.40 & 0.38 & 0.13 & 0.14 & 0.01 & 0.07 \\
				T10\_I8 & 54.02 & 421.28 & 14.56 & 34.61 & 80.04 & \textbf{0.00} &  7.85 & 5.72 & 3.51 & 2.71 & 3.75 & 7.50 & 6.60 & 0.57 & 1.39 & \textbf{0.65} & 2.00 \\
				T10\_I9 & 50.51 & 512.42 & 22.85 &  91.58 & 42.88 &  \textbf{0.00} &  5.29 & 7.21 & 5.91 & 2.22 & 2.20 & 18.14 & 3.90 & 2.99 & 2.43 & \textbf{0.22} & 0.31 \\
				T10\_I10 & 66.45 & 682.43 & 7.14 &  92.16 & 74.07 & \textbf{0.00} & 0.70 & 1.79 & 2.56 & 4.19 & 2.77 & 6.57 & 1.68 & 2.00 & 0.75 & \textbf{0.75} & 0.96 \\ \midrule
				Mean & 64.44 & -- & 12.63 &  50.67 &  67.71 & \textbf{0.00} & 6.24 & 5.95 & 3.91 & 3.22 & 2.35 & 8.30 & 2.71 & 3.12 & 1.36 & 1.34 & 1.11 \\ \bottomrule
			\end{tabular}%
		} 
		\begin{tablenotes}
			\footnotesize
			\item[] \textit{Note:} "--" indicates that no feasible solution is found within the time limit.
		\end{tablenotes}
	\end{threeparttable}
	\vspace{-10pt}
\end{table*}

For instances with 20 or more tasks, Gurobi fails to produce a feasible solution within the 3600s runtime.
Similarly, while OR-Tools identifies feasible solutions for $n=20$, it yields a high average RPD of 22.80\% and fails to find any valid solution for instances with 30 or more tasks.
Meanwhile, EDDBID performs poorly with significantly inferior results (mean RPD $>100\%$).
Consequently, these approaches are excluded from the subsequent comparison for large-scale instances.
To examine performance variations of metaheuristics across different instance scales, an analysis of variance is performed at a 95\% confidence level using Tukey’s Honest Significant Difference test.
The interaction plot is presented in Fig. \ref{interaction_plot}.
It is observed that the performance of all heuristic algorithms declines as the instance scale increases.
VNS, GA, DABC, and IG all show notable performance degradation on large-scale instances with 30 and 40 tasks, with mean RPDs exceeding 20\%.
Despite some fluctuations in performance, the proposed ALNS maintains mean RPDs around 10\% with no noticeable increase when $n = 20, 30, \textnormal{ and }40$.
Such consistent performance, achieved using the fixed parameter set, empirically validates ALNS's robustness against scale variations.
To further assess stability, RPD box-and-whisker plots are presented for instances with 10, 20, 30, and 40 tasks, as shown in Fig. \ref{box}.
The boxplots of our method show smaller boxes, lower medians, and shorter whiskers, suggesting improved performance and stability over its peers.

\begin{figure}[t]
	\centering
	\includegraphics[width=0.85\linewidth]{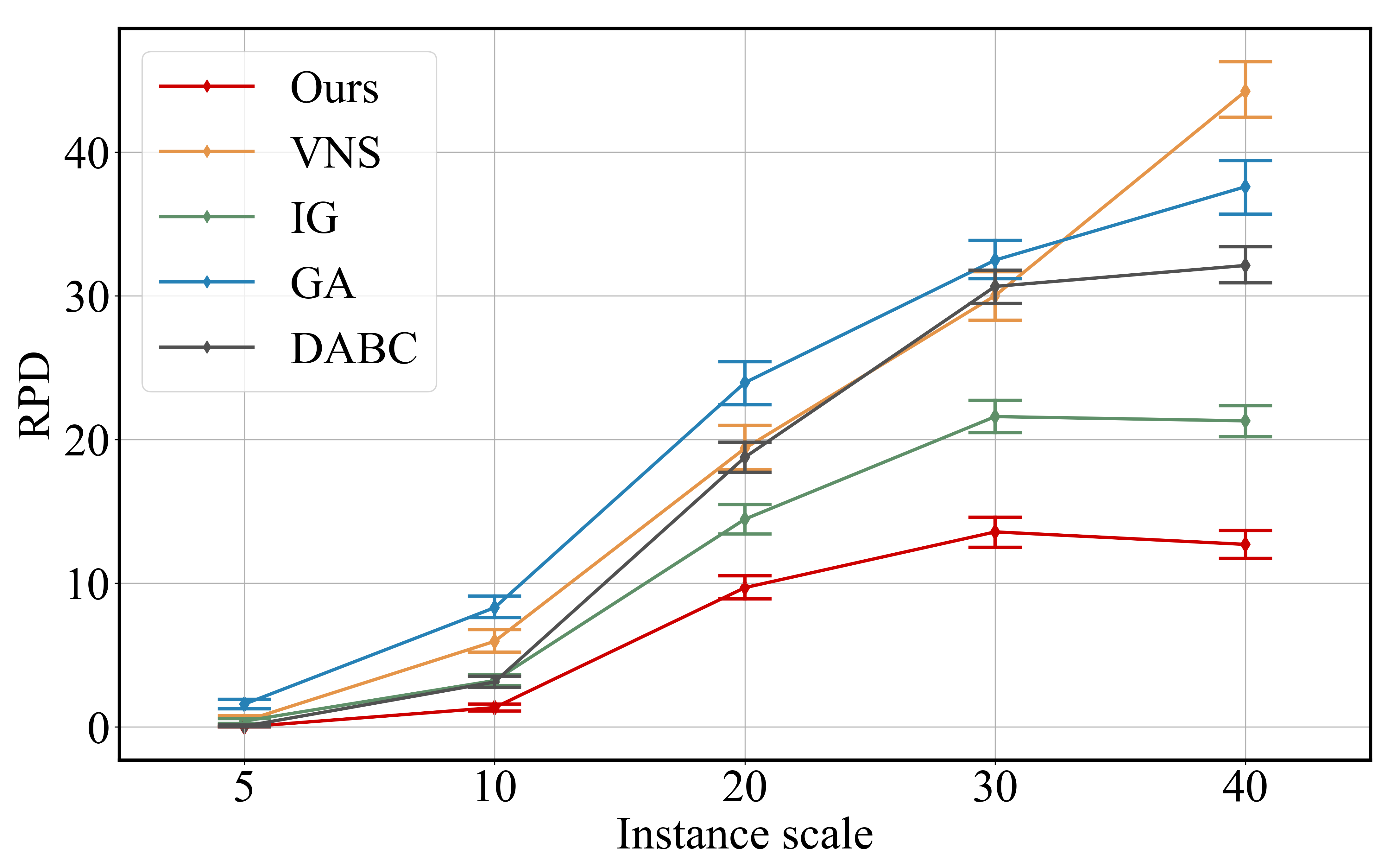}
	\caption{Interaction plot between five metaheuristics and instance scales.}
	\label{interaction_plot}
	\vspace{-10pt}
\end{figure}

\begin{figure}[ht]
	\centering \textbf{}
	\subfigure[$n=10$]{
		\includegraphics[width=0.43\linewidth]{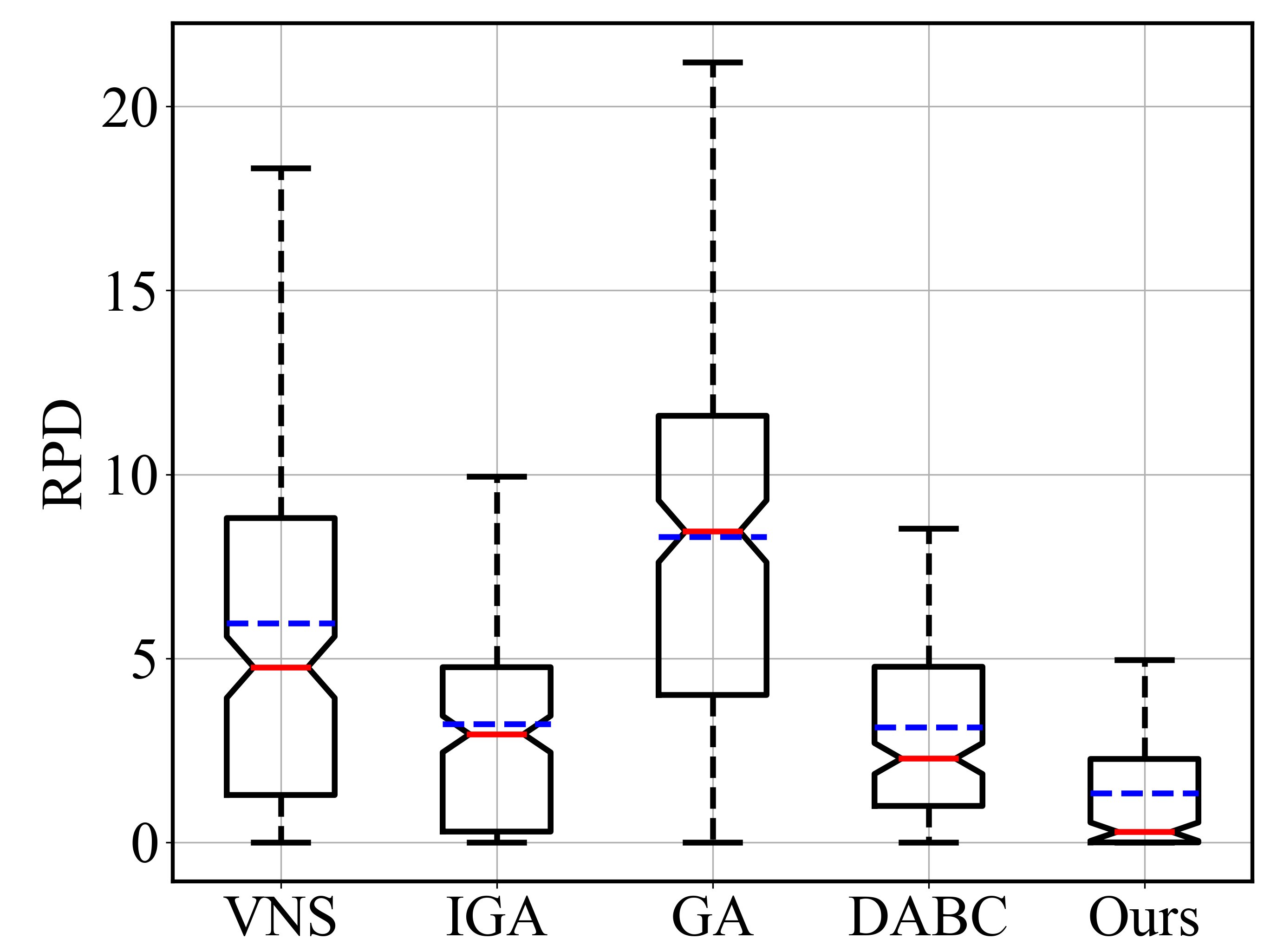}}
	\subfigure[$n=20$]{
		\includegraphics[width=0.43\linewidth]{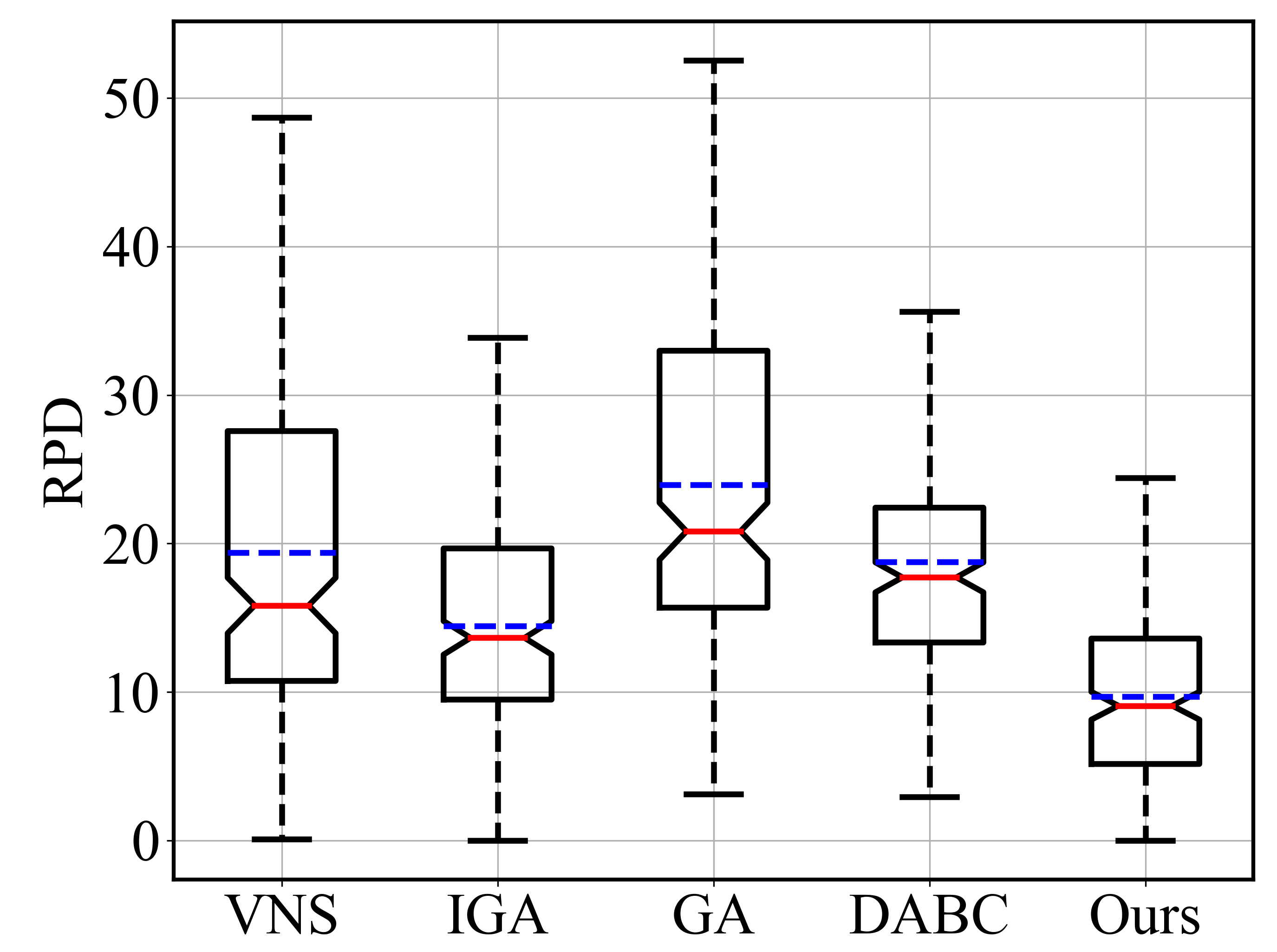}}
	\\
	\subfigure[$n=30$]{
		\includegraphics[width=0.43\linewidth]{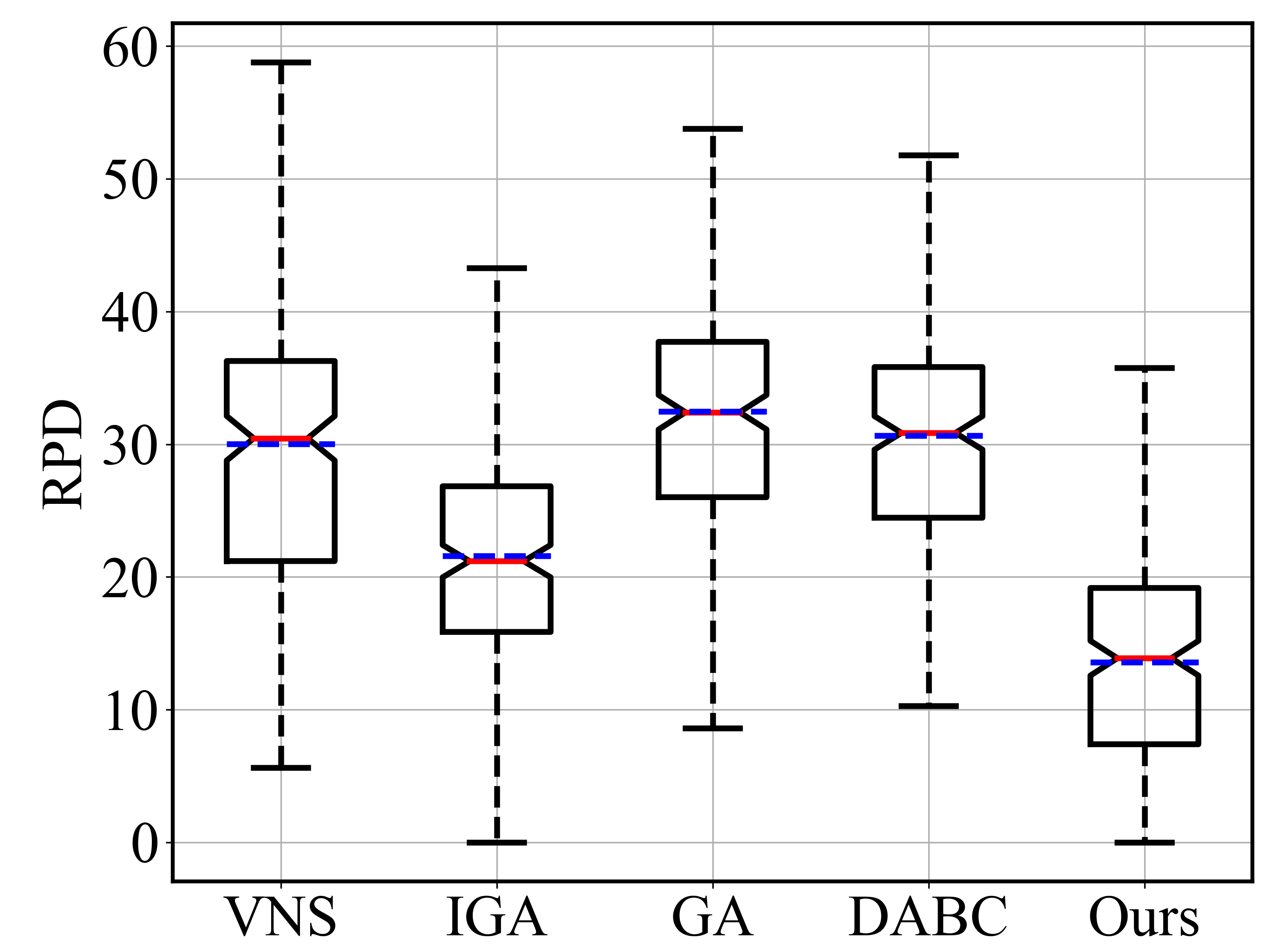}}
	\subfigure[$n=40$]{
		\includegraphics[width=0.43\linewidth]{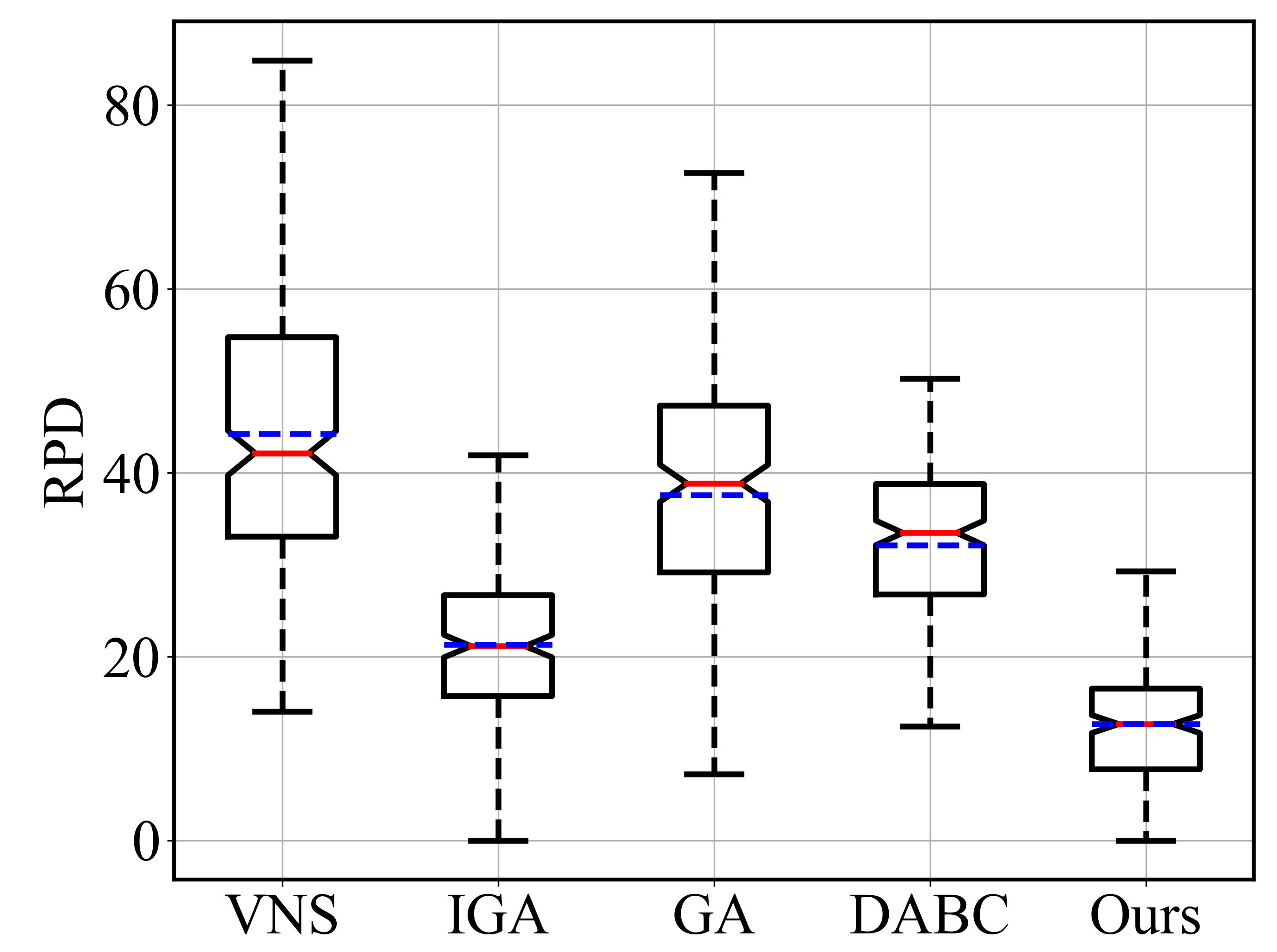}}
	\caption{Box-and-whisker plots of RPDs.}
	\label{box}
	\vspace{-10pt}
\end{figure}

\subsection{Operational Analysis and Managerial Insights}

To identify operational bottlenecks and support fleet configuration decisions, we first analyze the workload balance and resource utilization across different scales, as detailed in Table \ref{tab:operational_analysis_full}.
In it, $t_{\textnormal{W}}$ and $t_{\textnormal{S}}$ denote the average work and synchronization wait time (in seconds), while $\eta$ and $\gamma$ represent the utilization and wait ratios (in percentages), respectively.
Workload ratio is defined as the shuttle's effective work content divided by that of the carrier.
As shown in Table~\ref{tab:operational_analysis_full}, the workload ratio consistently remains around 0.7 across all scales, indicating that the work content required of shuttles is inherently lower than that of carriers.
This imbalance identifies the carrier as the natural bottleneck.
The operational time decomposition further supports this mechanism.
As the scale increases (e.g., $n=40$), carriers approach saturation ($\eta = 91.76\%$) with comparatively low synchronization wait time ($t_{\textnormal{S}} = 152.3$s).
Conversely, shuttles experience a substantial synchronization wait ($t_{\textnormal{S}} = 586.1$s), which is approximately half of their effective work time ($t_{\textnormal{W}} = 1156.5$s).
This pronounced disparity confirms that system throughput is constrained primarily by carrier capacity.

Building on these operational findings, we conduct a sensitivity analysis to optimize fleet configuration.
The carrier count $m^+$ is examined at five levels $\{2, \dots, 6\}$, and the shuttle count $m^-$ at nine levels $\{4, \dots, 12\}$.
Each combination is run 20 times on instance T30\_I1, with average travel distance and tardiness presented in the heatmaps in Fig. \ref{heatmap}.
The results confirm that the system is carrier-constrained under the workload ratio of 0.7.
As carriers are nearly saturated in the baseline C4\_S8 configuration, significant performance gains are achieved by increasing carrier capacity.
Specifically, increasing the carrier count by one (C5\_S8) drastically reduces tardiness from 634.58 to 76.18. 
In contrast, increasing the non-bottleneck resource, i.e., the shuttle (C4\_S9), results in a much smaller reduction to 416.53.
This confirms that when carrier capacity is the limiting factor, adding shuttles merely increases synchronization waits without significantly boosting productivity.
However, once the bottleneck is relieved ($m^+ \geq 5$), a diminishing marginal effect is observed.
To address the cost-benefit trade-off regarding fleet sizing, we analyze the marginal utility of further expansion.
Transitioning from C4\_S8 to C5\_S8 yields the highest return, resolving the primary bottleneck with an 88\% reduction in tardiness.
However, further expansion encounters diminishing returns or even negative ones.
Adding a sixth carrier (C6\_S8) reduces tardiness only marginally to 30.91, which may not justify the additional procurement cost.
Similarly, increasing the shuttle count to C5\_S9 results in a slight deterioration in performance due to increased coordination complexity.
In conclusion, C5\_S8 is identified as the most cost-effective configuration, striking an optimal balance between service quality and investment cost.

\begin{table}[htbp]
	\centering
	\caption{Operational workload and utilization analysis of carriers and shuttles}
	\label{tab:operational_analysis_full}
	\setlength{\tabcolsep}{2.5pt}
	\resizebox{\linewidth}{!}{%
		\begin{tabular}{cccccccccc}
			\toprule
			\multirow{2}[2]{*}{Scale} & \multicolumn{4}{c}{Carrier} & \multicolumn{4}{c}{Shuttle} & \multirow{2}[2]{*}{\shortstack{Workload\\Ratio}} \\
			\cmidrule(lr){2-5} \cmidrule(lr){6-9}
			& $t_{\textnormal{W}}$ & $t_{\textnormal{S}}$ & $\eta$ & $\gamma$ & $t_{\textnormal{W}}$ & $t_{\textnormal{S}}$ & $\eta$ & $\gamma$ &  \\
			\midrule
			5     & 183.5 & 19.9  & 71.68 & 7.72  & 128.2 & 23.2  & 41.75 & 13.70 & 0.70 \\
			10    & 402.1 & 63.9  & 86.74 & 13.26 & 279.1 & 69.8  & 57.25 & 19.05 & 0.69 \\
			20    & 840.9 & 101.8 & 89.22 & 10.78 & 571.0 & 253.2 & 64.95 & 30.38 & 0.68 \\
			30    & 1259.1 & 125.5 & 90.98 & 9.02  & 862.7 & 462.2 & 64.60 & 34.51 & 0.69 \\
			40    & 1700.4 & 152.3 & 91.76 & 8.24  & 1156.5 & 586.1 & 66.08 & 33.92 & 0.68 \\
			\bottomrule
		\end{tabular}%
	}
	\vspace{-10pt}
\end{table}

\begin{figure}[ht]
	\centering 
	\subfigure[Travel Distance]{
		\includegraphics[width=0.95\linewidth]{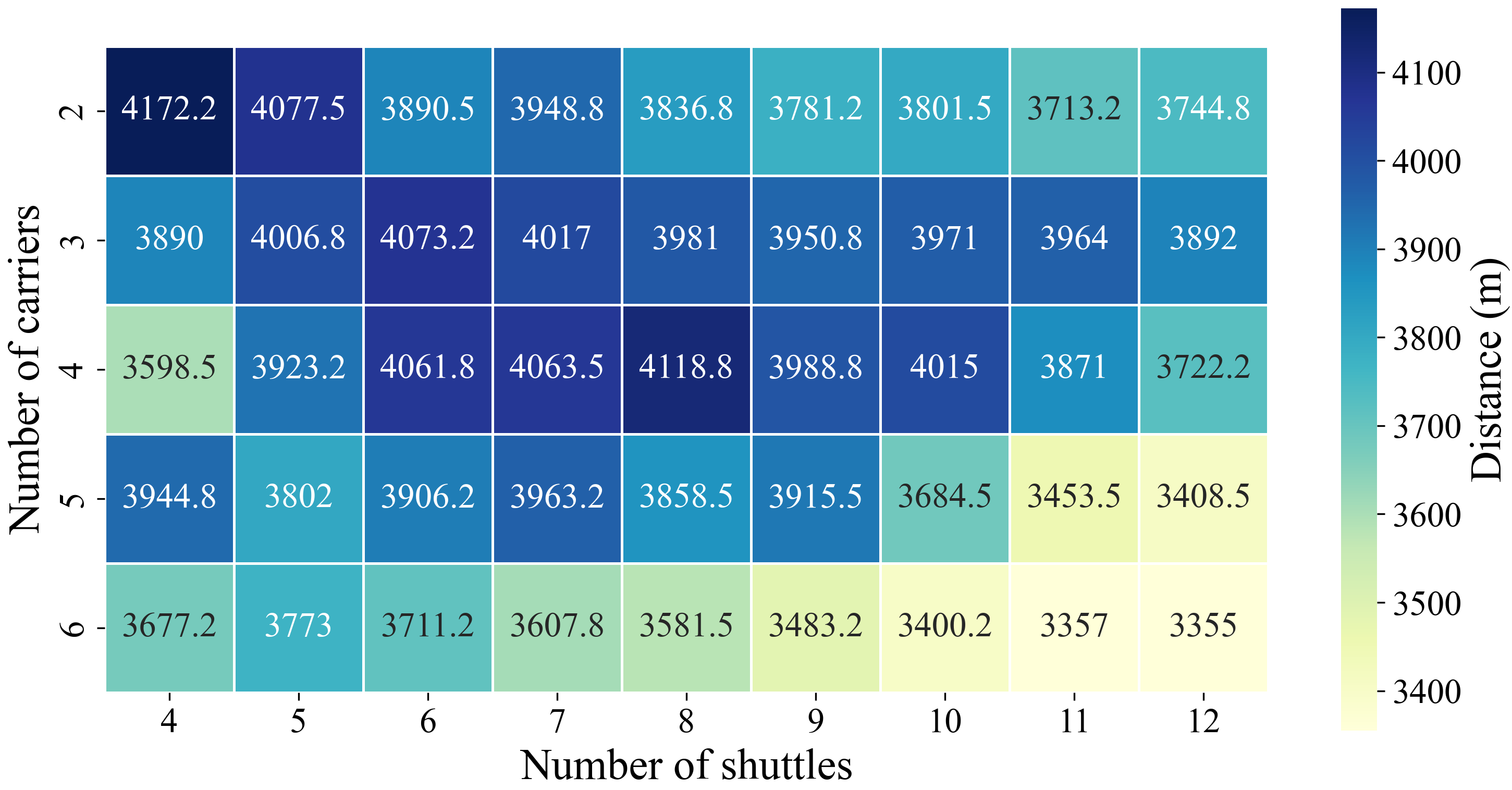}}
	\\
	\subfigure[Tardiness]{
		\includegraphics[width=0.95\linewidth]{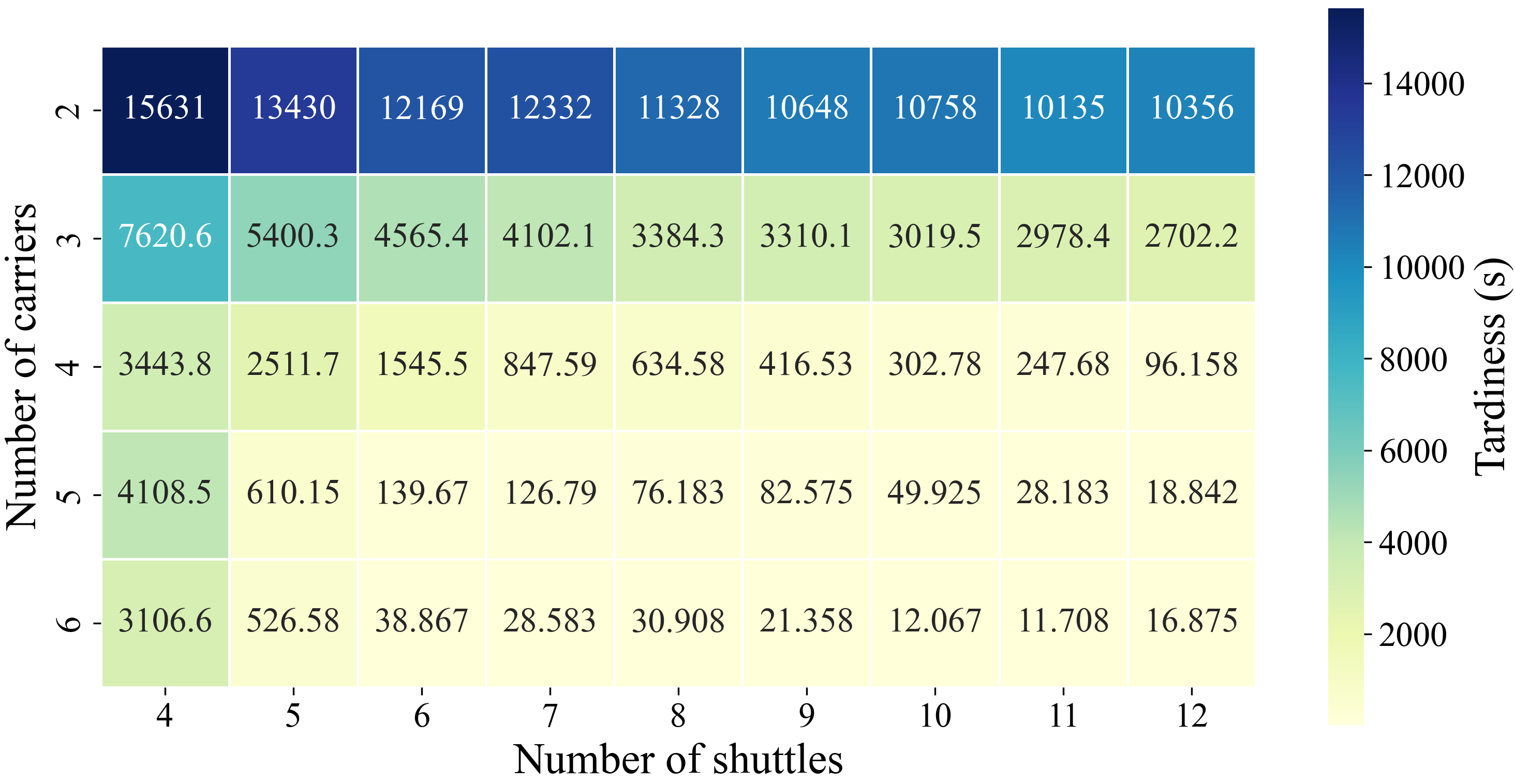}}
	\caption{Sensitivity heatmap for the configuration of carriers and shuttles.}
	\label{heatmap}
	\vspace{-10pt}
\end{figure}


\subsection{Generalizability and Scalability Analysis}
The proposed PN-based deadlock-free framework is not limited to the specific carrier--shuttle system or the industrial instance scales observed in the studied manufacturing system, but is applicable to a broader class of HRSP characterized by multi-agent synchronization.
Real-world examples include crane-based container terminals \cite{chen2020yard} and semiconductor transport systems \cite{chou2025multiagent}, where task execution necessitates the simultaneous collaboration of heterogeneous agents.
When applying the framework to such domains, the fundamental PN structure remains valid.
The main adaptation is to model a generic task requiring $k$ agents as a transition with $k$ input places.
Since the core logic of synchronization and coupled precedence is structurally invariant, the proposed deadlock-free theorems and algorithms remain applicable.

To further evaluate the scalability and generalizability of the proposed framework beyond the studied manufacturing system and its observed industrial scales, we extend the framework to a general HRSP with synchronization constraints \cite{fu2022robust}.
Synthetic HRSP datasets are generated with three robot types and fleet sizes of 2, 4, and 6.
Comparative results on larger-scale instances ranging from $n=10$ to $80$ are summarized in Table~\ref{tab:comparison_summary}.
Here, ALNS$_{\text{Base}}$ is the baseline ALNS without the PN modeling and BDS strategy, while ALNS$_{\text{PN}}$ is our proposed method.
As shown in Table~\ref{tab:comparison_summary}, incorporating the PN-based deadlock-free strategy substantially improves both solution quality and computational efficiency, especially for large-scale HRSP instances with tightly coupled dependencies.
Moreover, the benefit of the PN-based deadlock-free strategy becomes increasingly pronounced as the problem scale grows.
When the scale reaches $n=80$, the performance gap between $\textnormal{ALNS}_{\textnormal{PN}}$ and $\textnormal{ALNS}_{\textnormal{Base}}$ increases to 77.90\%, accompanied by a search efficiency improvement exceeding 350\%.
Furthermore, the developed ALNS demonstrates superior robustness compared to state-of-the-art metaheuristics.
While DABC and IG exhibit notable performance degradation on large-scale instances, our method consistently maintains the lowest RPDs.
Overall, these results validate both the scalability and the generalizability of the proposed framework, confirming the critical role of PN-based deadlock prevention in large-scale heterogeneous multi-robot scheduling with synchronization constraints.

\begin{table}[htbp]
	\centering
	\caption{Generalization comparison on the general HRSP under different problem scales.}
	\label{tab:comparison_summary}
	\renewcommand{\arraystretch}{1.1}
	\resizebox{\columnwidth}{!}{%
		\begin{tabular}{c cccc cc}
			\toprule
			\multirow{2}{*}{Scale} & \multicolumn{4}{c}{Avg. RPD (\%)} & \multicolumn{2}{c}{PN Improvement} \\
			\cmidrule(lr){2-5} \cmidrule(lr){6-7}
			& IG & DABC & ALNS$_{\text{Base}}$ & \textbf{ALNS$_{\text{PN}}$} & $\textnormal{G}_\textnormal{AP}$ (\%) & $\textnormal{A}_\textnormal{R}$ (\%) \\
			\midrule
			10 & 4.10 & 0.47 & 0.35 & \textbf{0.25} & 0.10 & 8.40 \\
			20 & 6.36 & 9.33 & 3.83 & \textbf{2.65} & 1.18 & 54.42 \\
			40 & 8.10 & 16.50 & 11.11 & \textbf{4.67} & 6.44 & 122.85 \\
			60 & 7.08 & 21.82 & 17.50 & \textbf{4.99} & 12.51 & 225.04 \\
			80 & 14.32 & 35.33 & 86.02 & \textbf{8.12} & 77.90 & 359.17 \\
			\bottomrule
		\end{tabular}%
	}
	\vspace{-5pt}
\end{table}

\section{Conclusion}\label{Conclusion}
This paper investigates a new scheduling problem involving attachable heterogeneous AGVs motivated by practical engineering applications.
To tackle the high combinatorial complexity and deadlock challenges induced by synchronized collaboration, we propose a PN-enhanced deadlock-free scheduling framework integrated into the ALNS algorithm.
PNs are constructed to map static permutation-based solutions into dynamic execution processes, enabling precise decoding via state evolution and proactive deadlock prevention through structural analysis.
Extensive experiments on real-world engineering data and synthetic general instances validate that the proposed framework significantly enhances computational efficiency and optimization performance.
Moreover, the developed ALNS outperforms the on-site scheduling policy, exact solvers, and state-of-the-art metaheuristics, particularly in solving large-scale problems.
Finally, operational analyses identify the carrier as the primary system bottleneck, yielding managerial insights for cost-effective fleet configuration.

Our future research plans to focus on two promising directions.
First, we aim to develop online, rolling-horizon rescheduling frameworks to enhance system robustness against stochastic uncertainties. 
Second, we plan to integrate learning-based techniques, such as DRL-guided operator selection and generative solution construction, to further enhance decision-making intelligence.

\bibliography{cas-refs}
\bibliographystyle{IEEEtran}

\clearpage

\section*{Supplementary Material}
\addcontentsline{toc}{section}{Supplementary Material}
\renewcommand{\theequation}{S.\arabic{equation}}
\setcounter{equation}{0}
\renewcommand{\thesection}{S.\Roman{section}}
\setcounter{section}{0}
\renewcommand{\thefigure}{S.\arabic{figure}}
\renewcommand{\thetable}{S.\arabic{table}}
\setcounter{figure}{0}
\setcounter{table}{0}
\section{Proof for Theorem~\ref{M_F_reach}} \label{proof_M_F_reach}


\vspace{5pt}
\noindent \textit{To prove Theorem \ref{M_F_reach}, we first introduce the following lemma:}

\begin{Lemma}
	In PN $\Sigma(\Pi)$, $\forall t \in \mathcal{T}$ can be fired at most once.
	\label{fire_once}
\end{Lemma}

\begin{proof}[\textbf{Proof of Lemma \ref{fire_once}}]
	The proof is established by contradiction. 
	Assume that there exists a transition $t_i \in \mathcal{T}$ that fires more than once.
	Based on the firing rule, the re-firing of $t_i$ necessitates token replenishment in its input places.
	Let $\pi_r$ be a route containing task $i$, and let $k$ denote the first task in $\pi_r$, i.e., $(0,k)\in \Pi^+$ or $(0,k)\in \Pi^-$.
	Since $\pi_r$ is a finite sequence starting from the virtual depot $0$, this requirement for repeated token supply recursively propagates backward, implying the re-firing of transition $t_k$ associated with the first task $k$.
	According to the $Post$ function defined in Definition \ref{definition_petri_net}, a token can be deposited into an input place $p_k^+$ or $p_k^-$ of $t_k$ only if there exists a predecessor $j$ such that $(j,k)\in \Pi^+$ or $(j,k)\in \Pi^-$.
	As $k$ is the first task of the route, its only predecessor in $\pi_r$ is depot $0$.
	Since the transition set is defined as $\mathcal{T} = \{t_a \mid a \in \mathcal{N}\}$ and depot $0 \notin \mathcal{N}$, no transition corresponding to node $0$ exists in $\Sigma(\Pi)$.
	Therefore, the input place associated with route $\pi_r$ has no input transitions, i.e., ${}^{\bullet}p_k^+ = \emptyset$ (if $r \in \mathcal{V}^+$) or ${}^{\bullet}p_k^- = \emptyset$ (if $r \in \mathcal{V}^-$).
	Since these places receive tokens exclusively from the initial marking $M_0$ and possess no input transitions for replenishment, $t_k$ can be fired at most once.
	This contradicts the implication that the dependency chain requires $t_k$ to re-fire.
	Thus, the assumption is invalid, and $\forall t \in \mathcal{T}$ can be fired at most once.
\end{proof}

\vspace{10pt}
\begin{proof}[\textbf{Proof of Theorem \ref{M_F_reach}}]
	The proof consists of two parts: necessity and sufficiency.
	
	\vspace{5pt}
	\noindent \textbf{1) Necessity}
	
	The necessity is proved by contradiction.
	Assume that $\Sigma(\Pi)$ reaches $M_\textnormal{F}$, yet $\exists t_k \in \mathcal{T}$ remains unfired.
	Consider the route $\pi_r$ containing task $k$. Let $z$ be the terminal task of this route.
	According to the $Post$ function in Definition \ref{definition_petri_net}, the virtual depot place $p_0^r$ receives a token exclusively from $t_z$ (i.e., $Post(t_z, p_0^r)=1$).
	Since $M_\textnormal{F}(p_0^r)=1$ by Definition \ref{terminal_marking}, $t_z$ must have fired.
	If $k=z$, this contradicts the assumption immediately.
	If $k \neq z$, we trace the task sequence in $\pi_r$ downstream from $k$.
	By Definition \ref{definition_petri_net}, for any consecutive tasks $(u, v) \in \pi_r$, the firing of $t_v$ strictly requires a token in $p_v^+$ (or $p_v^-$) supplied by $t_u$.
	Consequently, if a predecessor $t_u$ remains unfired, its successor $t_v$ is permanently disabled.
	By induction, the unfired status of $t_k$ propagates to $t_z$, implying $t_z$ cannot fire.
	This results in $M(p_0^r)=0$, which contradicts the definition of $M_\textnormal{F}$.
	
	Thus, all transitions must fire.
	Combined with Lemma \ref{fire_once}, we conclude that every transition fires exactly once.
	
	\vspace{10pt}
	\noindent \textbf{2) Sufficiency}
	
	Assume every transition $t \in \mathcal{T}$ fires exactly once to reach a marking $M^*$.
	For any place $p_0^r \in \mathcal{P}_0^+ \cup \mathcal{P}_0^-$, Definition \ref{definition_petri_net} dictates that $M_0(p_0^r)=0$ and that $p_0^r$ possesses a unique input transition but no output transition.
	Therefore, the single firing of its unique input transition deposits exactly one token, yielding $M^*(p_0^r)=1$.
	For any process place $p \in \mathcal{P}^+ \cup \mathcal{P}^-$, Definition \ref{definition_petri_net} designates a unique output transition $t$.
	If $M_0(p)=1$, $p$ possesses no input transition. 
	Thus, the firing of $t$ strictly consumes the initial token.
	If $M_0(p)=0$, $p$ possesses exactly one input transition.
	The unique predecessor fires exactly once, depositing a token which is subsequently consumed by the firing of $t$.
	In both cases, $M^*(p)=0$.
	By Definition \ref{terminal_marking}, the resulting marking $M^*$ satisfies the exact conditions of the final marking, proving $M^* = M_\textnormal{F}$.
\end{proof}

\section{Proof for Theorem~\ref{deadlock_feasible}} \label{proof_deadlock_feasibility}


\begin{proof}[\textbf{Proof of Theorem \ref{deadlock_feasible}}]
	The proof establishes the necessary and sufficient conditions separately.
	
	\vspace{5pt}
	\noindent \textbf{1) Necessity}
	
	The necessity is proved by contradiction.
	Assume that solution $\Pi$ is feasible, yet the PN $\Sigma(\Pi)$ encounters a deadlock at a reachable marking $M_\textnormal{D} \in \mathcal{R}(M_0) \setminus \{M_\textnormal{F}\}$.
	
	Let $\mathcal{T}_\textnormal{D} \subseteq \mathcal{T}$ denote the set of transitions that have not fired at $M_\textnormal{D}$.
	Since $M_\textnormal{D} \neq M_\textnormal{F}$, by Theorem~\ref{M_F_reach}, $\mathcal{T}_\textnormal{D} \neq \emptyset$.
	Because $M_\textnormal{D}$ is a deadlock marking, every transition in $\mathcal{T}_\textnormal{D}$ is disabled.
	Consequently, for any $t_j \in \mathcal{T}_\textnormal{D}$, there exists at least one input place $p \in {}^{\bullet}t_j$ such that $M_\textnormal{D}(p) = 0$.
	If $p$ were a source place (i.e., $M_0(p)=1$), the initial token would remain in $p$ because $t_j$ has not fired. 
	This implies $M_\textnormal{D}(p)=1$, contradicting the fact that $p$ is empty.
	Therefore, $p$ must be an intermediate place with $M_0(p)=0$. 
	By Definition~\ref{definition_petri_net}, such a place possesses a unique input transition, denoted as $t_i \in {}^{\bullet}p$.
	
	Now, consider the state of this predecessor $t_i$.
	If $t_i$ had already fired, it would have deposited a token into $p$. Since $t_j$ has not fired to consume it, $p$ would hold a token ($M_\textnormal{D}(p)=1$), which again contradicts the emptiness of $p$.
	Thus, the predecessor $t_i$ must also be unfired, i.e., $t_i \in \mathcal{T}_\textnormal{D}$.
	By Definition~\ref{definition_petri_net}, the relationship $t_i \in {}^{\bullet}p$ and $p \in {}^{\bullet}t_j$ implies a precedence relation $(i,j) \in \Pi$, i.e., $i \prec j$.
	Since no task can be its own immediate predecessor in the route, $t_i \neq t_j$.
	
	In summary, for every $t_j \in \mathcal{T}_\textnormal{D}$, there exists a distinct predecessor $t_i \in \mathcal{T}_\textnormal{D}, t_i \neq t_j$ such that $i \prec j$.
	Since $\mathcal{T}_\textnormal{D}$ is a finite set, tracing these precedence dependencies must eventually revisit a transition (by the pigeonhole principle), forming a closed loop $t_{k} \prec \dots \prec t_{k}$.
	This implies the existence of a circular dependency (circular wait) in $\Pi$, contradicting the assumption that $\Pi$ is a feasible solution.
	Therefore, if $\Pi$ is feasible, $\Sigma(\Pi)$ must be deadlock-free at any reachable marking $M \in \mathcal{R}(M_0) \setminus \{M_\textnormal{F}\}$.
	
	\vspace{10pt}
	\noindent \textbf{2) Sufficiency}
	
	The sufficiency is proved by contraposition. 
	Assuming solution $\Pi$ is infeasible, we demonstrate that there exists a reachable marking $M_\textnormal{D} \in \mathcal{R}(M_0) \setminus \{M_\textnormal{F}\}$ that constitutes a deadlock.
	
	The infeasibility of $\Pi$ implies the existence of a set of tasks $\mathcal{N}_\mathcal{C} \subseteq \mathcal{N}$ forming a circular precedence relationship.
	By Definition~\ref{definition_petri_net}, the chain of precedence relations in this cycle maps to a structural loop of places and transitions in $\Sigma(\Pi)$, denoted as $\mathcal{C}$.
	Let $\mathcal{P}_\mathcal{C}$ and $\mathcal{T}_\mathcal{C}$ denote the sets of places and transitions constituting this loop, respectively.
	Consider any place $p \in \mathcal{P}_\mathcal{C}$.
	By the definition of a loop, $p$ connects two transitions $t_i, t_j \in \mathcal{T}$.
	Let $t_i$ denote the unique input transition of $p$ and $t_j$ denote its output transition, implying $i \prec j$.
	
	By Definition \ref{definition_petri_net}, the transition set is defined as $\mathcal{T} = \{ t_k \mid k \in \mathcal{N} \}$.
	Since the virtual depot $0 \notin \mathcal{N}$, the existence of $t_i \in \mathcal{T}$ strictly implies $i \neq 0$.
	According to the definition of $M_0$ in Definition~\ref{definition_petri_net}, tokens are initially assigned exclusively to places connected to the depot (i.e., where $(0, i) \in \Pi$).
	Consequently, $p$ receives no initial token, i.e., $M_0(p)=0$.
	It follows that the loop $\mathcal{C}$ holds no tokens at $M_0$, satisfying $M_0(p) = 0$ for all $p \in \mathcal{P}_\mathcal{C}$.
	
	Structurally, within the closed loop $\mathcal{C}$, the firing of any transition $t \in \mathcal{T}_\mathcal{C}$ strictly requires a token from its specific input place $p \in \mathcal{P}_\mathcal{C}$.
	However, places in $\mathcal{P}_\mathcal{C}$ can receive tokens exclusively from the firing of predecessor transitions within the same loop $\mathcal{C}$.
	Given that $\mathcal{C}$ initially contains no tokens ($M_0(p)=0, \forall p \in \mathcal{P}_\mathcal{C}$), the enabling condition for any transition in $\mathcal{T}_\mathcal{C}$ can never be satisfied.
	Consequently, every transition $t \in \mathcal{T}_\mathcal{C}$ remains permanently disabled.
	
	Since the transitions in $\mathcal{T}_\mathcal{C}$ never fire, the condition of Theorem~\ref{M_F_reach} ($\forall t \in \mathcal{T}$ fires exactly once) is violated, rendering $M_\textnormal{F}$ unreachable.
	By Lemma \ref{fire_once}, since each transition fires at most once, the system is free of livelocks and the state evolution must eventually terminate.
	Given that $M_\textnormal{F}$ is unreachable, this termination state must be a deadlock marking $M_\textnormal{D} \in \mathcal{R}(M_0) \setminus \{M_\textnormal{F}\}$.
\end{proof}

\section{Proof for Algorithm 2} \label{Algorithm_proof}

\begin{proof}[\textbf{Proof of Algorithm 2}]
	The proof establishes the soundness and completeness of the BDS algorithm.
	
	\vspace{5pt}
	\noindent \textbf{1) Soundness}
	
	To prove soundness, we demonstrate that any position identified by BDS leads to a circular dependency if selected.
	
	Consider any index $\overleftarrow{\epsilon}^-(b) \in \mathcal{L}_{\textnormal{B}}$ identified by the backward search.
	By the design of the breadth-first search in Algorithm 2, the inclusion of $b$ implies the existence of a directed path from transition $t_b$ to $t_k$ in $\Sigma(\Pi)$.
	Structurally, this path signifies an existing precedence relation $b \prec \dots \prec k$.
	Inserting task $k$ at position $\overleftarrow{\epsilon}^-(b)$ (i.e., immediately preceding $b$ in $\Pi^-$) imposes a new direct precedence $k \prec b$.
	Combining these relations yields a closed loop $b \prec \dots \prec k \prec b$, constituting a circular wait.
	Thus, position $\overleftarrow{\epsilon}^-(b)$ is infeasible.
	Symmetrically, for any index $\overrightarrow{\epsilon}^-(f) \in \mathcal{L}_{\textnormal{F}}$ identified by the forward search, there exists a directed path from $t_k$ to $t_f$, implying $k \prec \dots \prec f$.
	Inserting task $k$ at $\overrightarrow{\epsilon}^-(f)$ (i.e., immediately succeeding $f$) imposes $f \prec k$.
	This forms a closed loop $k \prec \dots \prec f \prec k$, rendering the position infeasible.
	
	Consequently, every position in $\mathcal{L}_{\textnormal{B}} \cup \mathcal{L}_{\textnormal{F}}$ leads to a circular wait, proving the soundness of BDS.
	
	\vspace{10pt}
	\noindent \textbf{2) Completeness}
	
	To prove completeness, we demonstrate that any infeasible insertion position is necessarily identified by BDS.
	
	Assume the insertion of task $k$ at position $\overleftarrow{\epsilon}^-(b)$ is infeasible.
	This implies that the introduced precedence $k \prec b$ completes a circular dependency.
	Given that the solution prior to insertion was feasible (acyclic), the remainder of this cycle must consist of a pre-existing chain of dependencies from $b$ to $k$, i.e., $b \prec \dots \prec k$.
	In PN $\Sigma(\Pi)$, this dependency maps to a directed structural path from transition $t_b$ to $t_k$.
	Since BDS employs a backward breadth-first search to exhaustively traverse all predecessors of $t_k$, it inevitably visits $t_b$ and appends $\overleftarrow{\epsilon}^-(b)$ to the set $\mathcal{L}_{\textnormal{B}}$.
	Symmetrically, if inserting at $\overrightarrow{\epsilon}^-(f)$ is infeasible, the induced precedence $f \prec k$ closes a cycle with an existing chain $k \prec \dots \prec f$, corresponding to a directed path from $t_k$ to $t_f$ in $\Sigma(\Pi)$.
	The forward breadth-first search in BDS guarantees the identification of $t_f$, thereby including $\overrightarrow{\epsilon}^-(f)$ in $\mathcal{L}_{\textnormal{F}}$.
	
	Consequently, BDS identifies all infeasible positions, establishing its completeness.
\end{proof}

\section{Firing process}
The FDD process for the example solution $\Pi$ given in Section IV-A is represented by the firing sequence $M_0 [ t_5 \rangle M_1 [ t_6 \rangle M_2 [ t_1 \rangle M_3 [ t_4 \rangle M_4 [ t_3 \rangle M_5 [ t_7 \rangle M_6 [ t_2 \rangle M_\textnormal{F}$.
The firing process of PN $\Sigma(\Pi)$ is visualized in Fig. \ref{firing_process}.

\begin{figure*}[ht]
	\centering
	\includegraphics[width=0.95\linewidth]{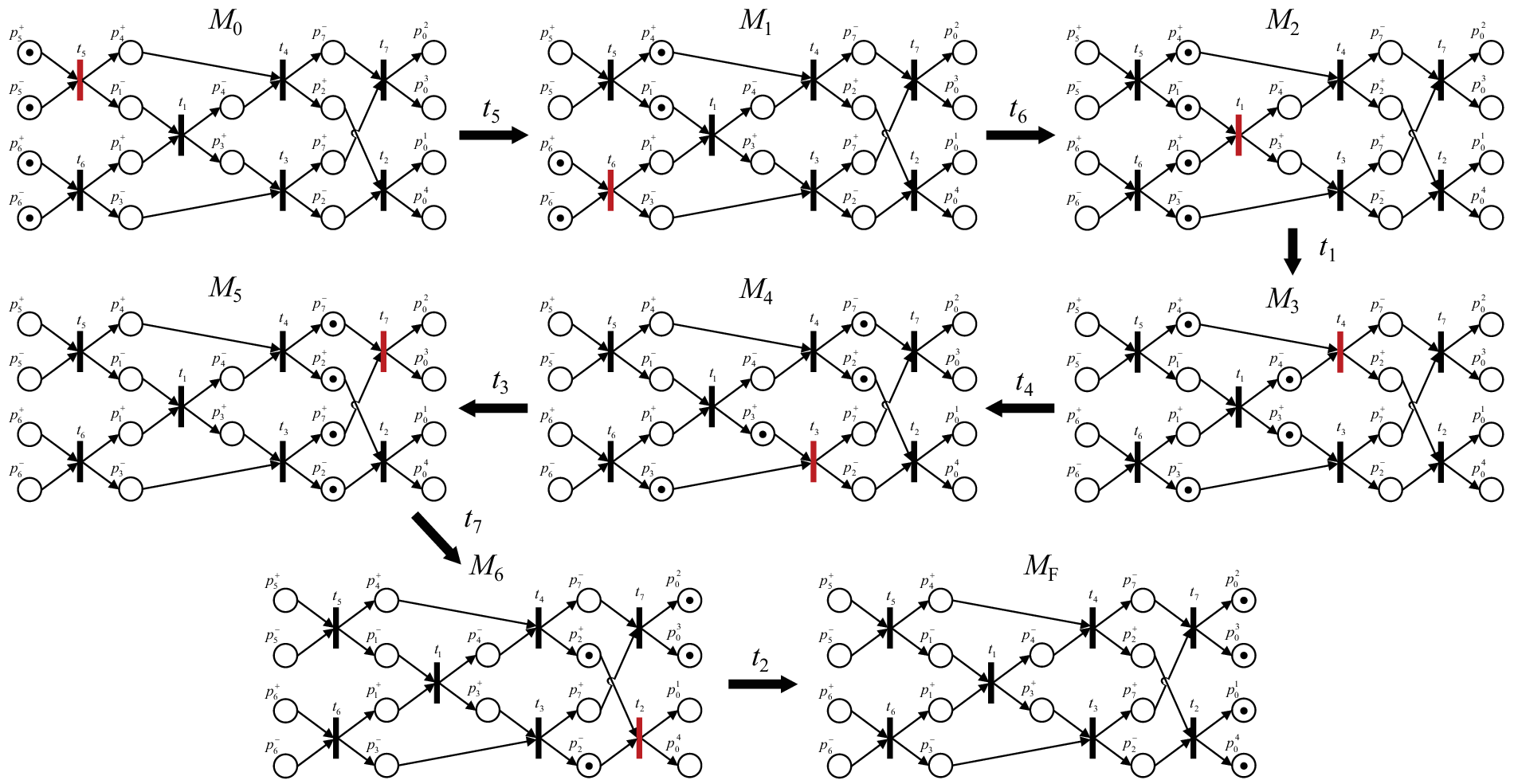}
	\caption{Firing process of PN $\Sigma(\Pi)$.}
	\label{firing_process}
\end{figure*}

\end{document}